\documentclass{article}
\usepackage{amssymb}
\usepackage{amsmath}

\newtheorem{theorem}{Theorem}

\newtheorem{corollary}[theorem]{Corollary}

\newtheorem{definition}[theorem]{Definition}
\newtheorem{example}[theorem]{Example}

\newtheorem{proposition}[theorem]{Proposition}
\newtheorem{remark}[theorem]{Remark}

\newenvironment{proof}[1][Proof]{\noindent\textbf{#1.} }{\ \rule{0.5em}{0.5em}}
\input{tcilatex}
\begin{document}
\title{(Pseudo) Generalized Kaluza-Klein $G$-Spaces and Einstein Equations}
\author{C. M. Arcu\c{s} and E. Peyghan}
\maketitle

\maketitle
\begin{abstract}
Introducing the Lie algebroid generalized tangent bundle of a Kaluza-Klein bundle, we
develop the theory of general distinguished linear connections for this space. In particular, using the Lie algebroid generalized tangent bundle of the Kaluza-Klein vector
bundle, we present the $\left( g,h\right) $-lift of a curve on the base $M$ and we  characterize the horizontal and vertical parallelism of the $\left( g,h\right) $-lift of accelerations with respect to a distinguished linear $\left( \rho ,\eta \right) $-connection. Moreover, we study the torsion, curvature and Ricci
tensor field associated to a distinguished linear $\left( \rho ,\eta \right)
$-connection and we obtain the identities of Cartan and Bianchi type in the
general framework of the Lie algebroid generalized tangent bundle of a
Kaluza-Klein bundle. Finally, we introduce the theory of (pseudo) generalized
Kaluza-Klein G-spaces and we develop the Einstein equations in this general framework.
\end{abstract}
\section{Introduction}
Recently, Lie algebroids are important issues in physics and mechanics since the extension of Lagrangian and Hamiltonian systems to their entity \cite{f, LMM, Pop, W} and catching the poisson  structure \cite{popescu0}. Then, Arcu\c{s} introduced generalized Lie algebroids as the extension of Lie algebroids and he studied geometrical and physical concepts for these spaces \cite{A0, A, A2, A1}. Indeed a generalized Lie algebroid is a extension of Lie algebroid from one base manifold to a pair of diffeomorphic base manifolds.

The Einstein theory of general relativity and Maxwell theory of
electromagnetism was independent developed in the world of physicists. In \cite{K} Kaluza proposed to unify these two theories using a
five-dimensional manifold. Kaluza's achievement was possible if the
components of the pseudo-Riemannian metric on the five-dimensional manifold
does not depend on the fifth coordinate (cylinder condition). Then, Klein \cite{Ke} added the
condition of compactification, namely the space is closed by a very
small circle in the direction of the fifth dimension. Therefore, the Kaluza-Klein theory emerged to unify of the Einstein theory of general
gravity and Maxwell theory of electromagnetism.

In \cite{EB}, Einstein and Bergmann proposed a first
generalization of Kaluza-Klein theory using a pseudo-Riemannian metric such
that its components are periodic in the fifth coordinate. So, the cylinder
condition was partially satisfied and, for the firs time, a covariant
derivative of a vector bundle over the five-dimensional manifold was
introduced.

A well known generalization of the five-dimensional Kaluza-Klein theory is
the so called the space time matter (STM) theory. Cosmological solutions in
which both the cylinder condition and the compactification condition were
removed are presented in \cite{C, LW, P, PW}. An
excelent survey on STM theory is presented in the paper of Overduin and Wesson
\cite{OW}.

The idea to construct exact solutions with Lie and Clifford algebroid symmetries  in modified and extra dimension gravity and matter field theories was elaborated originally in a series of preprints by S. Vacaru \cite{V1, V2, V3}; see further developments and reviews of results on nonholonomic algebroids  and Einstein-Direact structures, Finsler-Lagrange-Hamilton algebroid spaces and geometric flows on Lie algebroid in \cite{V4, V5, V6}.

Recently, Bejancu developed a new point of
view on a general Kaluza-Klein theory using the product manifold $\overline{M}%
=M\times K,$ where $M$ is a four-dimensional manifold and $K$ is a one-dimensional manifold \cite{B0, B, B1}. The tangent bundle of $\overline{M}$ is the space
used to develop the theory as a direct sum between the horizontal
distribution $H\overline{M}$ and vertical distribution $V\overline{M}.$ The
novelty is determined by the Riemannian horizontal connection which plays
the same role as Levi-Civita linear connection on the space time manifold $M$
for the classical Kaluza-Klein theory. Using this linear connection, which is not distinguished connection, Bejancu introduced the
Einstein gravitational tensor field and he write the field equations on $%
\left( \overline{M},\overline{g}\right) .$ A new method for the study of
general higher dimension Kaluza-Klein theory is presented by Bejancu in \cite{B2}.

We remark that Bejancu used the usual Lie algebroid tangent bundle. The purpose of this paper is to develop a general method to study of the Kaluza-Klein theory using the generalized connection theory presented by Arcu\c{s} in \cite{A0, A, A2, A1}.

This paper is arranged as follows. In Sec. 3, using a (generalized) Lie algebroid
presented in Sec. 2, we introduce the Lie algebroid generalized tangent
bundle of a Kaluza-Klein bundle and we study the connections on a Kaluza-Klein bundle. Moreover, we
develop the theory of general distinguished linear connections for the Lie
algebroid generalized tangent bundle.
In Sec. 4, using the Lie algebroid generalized tangent bundle of the Kaluza-Klein vector
bundle, we present the $\left( g,h\right) $-lift of a curve on the base $M$. A characterization of horizontality (respectively, verticality)
parallelism of the $\left( g,h\right) $-lift of accelerations with respect to
a distinguished linear $\left( \rho ,\eta \right) $-connection is presented.
Sec. 5 is dedicated to study the torsion and curvature of a
distinguished linear $\left( \rho ,\eta \right) $-connection. Moreover, the Ricci
tensor field associated to a distinguished linear $\left( \rho ,\eta \right)
$-connection is introduced and identities of Cartan and Bianchi type in the
general framework of the Lie algebroid generalized tangent bundle of a
Kaluza-Klein bundle are presented. Finally, in Sec. 6, we present the theory of (pseudo) generalized
Kaluza-Klein G-spaces. A lot of examples of Kaluza-Klein G-spaces are
presented and the Einstein equations are developed in this general framework.
In particular, using the identities morphisms, the
usual Kaluza-Klein theory used by Bejancu is obtained, but the difference is that the
metrical linear connection used in our paper is distinguished connection.

\section{Preliminaries}
Let $(F,\nu ,N)$ be a vector bundle, $\Gamma ( F, \nu, N)$ be the set of the sections of it and $\mathcal{F}(
N)$ be the smooth real-valued functions on $N$. Then $(\Gamma
(F,\nu ,N) ,+,\cdot)$ is a $\mathcal{F}(N)$-module. If $(\varphi ,\varphi _{0})$ is a morphism from $( F,\nu ,N)$ to $(F', \nu', N')$ such that $\varphi _{0}$ is a isomorphism from $N$ to $N'$, then using the operation
\begin{equation*}
\begin{array}{ccc}
\mathcal{F}\left( N\right) \times \Gamma \left(F', \nu', N'\right) & ^{\underrightarrow{~\ \ \cdot ~\ \ }} & \Gamma \left(
F', \nu', N'\right),\\
\left( f,u^{\prime }\right) & \longmapsto & f\circ \varphi _{0}^{-1}\cdot
u^{\prime },
\end{array}%
\end{equation*}%
it results that $(\Gamma( F', \nu', N') ,+,\cdot)$ is a $\mathcal{F}(N)$-module and we
obtain the modules morphism
\begin{equation*}
\begin{array}{ccc}
\Gamma \left(F,\nu ,N\right) & ^{\underrightarrow{~\ \ \Gamma \left(
\varphi ,\varphi _{0}\right) ~\ \ }} & \Gamma \left( F', \nu', N'\right),\\
u & \longmapsto & \Gamma \left( \varphi ,\varphi _{0}\right) u,
\end{array}%
\end{equation*}%
defined by
\begin{equation*}
\begin{array}{c}
\Gamma \left( \varphi ,\varphi _{0}\right) u\left( y\right) =\varphi \left(
u_{\varphi _{0}^{-1}\left( y\right) }\right) =\left( \varphi \circ u\circ
\varphi _{0}^{-1}\right) \left( y\right) ,%
\end{array}%
\end{equation*}%
for any $y\in N'$.
\begin{definition}\label{AP}
A generalized Lie algebroid is a vector bundle $(F,\nu ,N)$ given by the diagrams:
\begin{equation}
\begin{array}{c}
\begin{array}[b]{ccccc}
\left( F,\left[ ,\right] _{F,h}\right) & ^{\underrightarrow{~\ \ \ \rho \ \
\ \ }} & \left( TM,\left[ ,\right] _{TM}\right) & ^{\underrightarrow{~\ \ \
Th\ \ \ \ }} & \left( TN,\left[ ,\right] _{TN}\right)\\
~\downarrow \nu &  & ~\ \ \downarrow \tau _{M} &  & ~\ \ \downarrow \tau _{N}
\\
N & ^{\underrightarrow{~\ \ \ \eta ~\ \ }} & M & ^{\underrightarrow{~\ \ \
h~\ \ }} & N
\end{array}
\end{array}
\end{equation}
where $h$ and $\eta $ are arbitrary isomorphisms, $(\rho, \eta)$ is a vector bundles morphism from $(F,\nu, N)$ to $(TM,\tau _{M},M)$ and
\begin{equation*}
\begin{array}{ccc}
\Gamma \left( F,\nu ,N\right) \times \Gamma \left( F,\nu ,N\right) & ^{%
\underrightarrow{~\ \ \left[ ,\right] _{F,h}~\ \ }} & \Gamma \left( F,\nu
,N\right),\\
\left( u,v\right) & \longmapsto & \ \left[ u,v\right] _{F,h},
\end{array}
\end{equation*}
is an operation satisfies in
\begin{equation*}
\begin{array}{c}
\left[ u,f\cdot v\right] _{F,h}=f\left[ u,v\right] _{F,h}+\Gamma \left(
Th\circ \rho ,h\circ \eta \right) \left( u\right) f\cdot v,\ \ \ \forall f\in \mathcal{F}(N),
\end{array}
\end{equation*}
such that

1) The 4-tuple $(\Gamma( F, \nu, N) ,+,\cdot, [ , ] _{F,h})$ is a Lie $\mathcal{F}(N)$-algebra,

2) The modules morphism $\Gamma(Th\circ\rho, h\circ \eta)$ is a Lie algebras morphism from
$
(\Gamma(F,\nu ,N), +, \cdot, [ , ] _{F,h})
$
to
$
(\Gamma(TN, \tau_{N}, N), +, \cdot, [ , ]_{TN}).
$
\end{definition}
We denote by $\Big((F, \nu, N), [ , ] _{F,h}, (\rho, \eta) \Big)$ the generalized Lie algebroid defined in the above. Moreover, the couple $\Big([ , ]
_{F,h}, (\rho, \eta)\Big)$ is called the \emph{generalized
Lie algebroid structure.}

A morphism from
$
\Big(( F,\nu ,N), [ , ] _{F,h}, (\rho, \eta)\Big)
$
to
$
\Big(( F^{\prime }, \nu ^{\prime }, N^{\prime }), [ , ]
_{F^{\prime }, h^{\prime }}, (\rho ^{\prime }, \eta ^{\prime })\Big)
$
is a morphism $(\varphi ,\varphi _{0})$ from $(F,\nu, N)$ to $(F',\nu', N')$ such that $\varphi _{0}$ is an isomorphism from $N$ to $N'$, and the modules morphism $\Gamma(\varphi, \varphi _{0})$
is a Lie algebras morphism from
$
\left( \Gamma \left( F,\nu ,N\right) ,+,\cdot ,\left[ ,\right] _{F,h}\right)
$
to
$
\left( \Gamma \left( F^{\prime },\nu ^{\prime },N^{\prime }\right) ,+,\cdot
, \left[ ,\right] _{F^{\prime },h^{\prime }}\right).
$
Thus, we can discuss about the category of
generalized Lie algebroids.
\begin{remark}
In the particular case, $\left( \eta ,h\right) =\left(
Id_{M},Id_{M}\right)$, we obtain the definition of Lie algebroid.
\end{remark}
If we take local coordinates $(x^i)$ and $(\chi^{\tilde{\imath}})$ on open sets $U\subset M$ and $V\subset N$, respectively, then we have the corresponding local coordinates $(x^i, y^i)$ and  $(\chi^{\tilde{\imath}}, z^{\tilde{\imath}})$ on $TM$ and $TN$, respectively, where $i,\tilde{\imath}\in 1, \ldots, m$. Moreover, a local basis $\{t_\alpha\}$ of the sections of $\nu^{-1}(V)\rightarrow V$ generates local coordinates $(\chi^{\tilde{\imath}}, z^{\alpha})$ on $F$, where $\alpha \in 1, \ldots, p$. If we consider the another local coordinates $(x^{i'}(x^i), y^{i'}(x^i, y^i))$, $(\chi^{{\tilde{\imath}}'}(\chi^{\tilde{\imath}}), z^{{\tilde{\imath}}'}(\chi^{\tilde{\imath}}, z^{\tilde{\imath}}))$ and $(\chi^{{\tilde{\imath}}'}(\chi^{\tilde{\imath}}), z^{{\alpha}'}(\chi^{\tilde{\imath}}, z^{{\alpha}}))$ on $TM$, $TN$ and $F$, respectively, then we have the following corresponding changes of coordinates
\begin{equation}\label{AP1}
z^{{\alpha}'}=\Lambda _{\alpha }^{{\alpha}'}z^\alpha,\ \ y^{i'}=\frac{\partial x^{i'}}{\partial x^i}y^i,\ \ z^{\tilde{\imath}\prime }=\frac{\partial \chi ^{\tilde{\imath}\prime }}{%
\partial \chi ^{\tilde{\imath}}}z^{\tilde{\imath}}.
\end{equation}
We assume that $\left( \theta ,\mu \right)=\left( Th\circ
\rho ,h\circ \eta \right) $. If $z^{\alpha }t_{\alpha }$ is a section of $\left(
F,\nu ,N\right)$, then for any $f\in \mathcal{F}\left( N\right) $ and $\varkappa \in N$ we have
\begin{equation*}
\begin{array}[t]{l}
\displaystyle%
\begin{array}{c}
\Gamma \left( Th\circ \rho ,h\circ \eta \right) \left( z^{\alpha }t_{\alpha
}\right) f\left( h\circ \eta \left( \varkappa \right) \right)
=\left( \theta _{\alpha }^{\tilde{\imath}}z^{\alpha }\frac{\partial f}{%
\partial \varkappa ^{\tilde{\imath}}}\right) \left( h\circ \eta \left(
\varkappa \right) \right)\\
=\left( \left( \rho _{\alpha }^{i}\circ h\right)
\left( z^{\alpha }\circ h\right) \frac{\partial f\circ h}{\partial x^{i}}%
\right) \left( \eta \left( \varkappa \right) \right) ,%
\end{array}%
\end{array}%
\end{equation*}%
where $\rho^i_\alpha$ are local functions on $N$ \cite{A0}. Now, we put $[t_\alpha, t_\beta]_{F, h}=L^\gamma_{\alpha\beta}t_\gamma$, where $L^\gamma_{\alpha\beta}$ are local functions on $N$ and $\alpha, \beta, \gamma\in 1, \ldots, p$. It is easy to see that $L^\gamma_{\alpha\beta}=-L^\gamma_{\beta\alpha}$. Moreover, the condition (2) of {\it Definition \ref{AP}} implies that
\begin{equation*}
\begin{array}{c}
\displaystyle\left( L_{\alpha \beta }^{\gamma }\circ h\right) \left( \rho
_{\gamma }^{k}\circ h\right) =\left( \rho _{\alpha }^{i}\circ h\right) \frac{%
\partial \left( \rho _{\beta }^{k}\circ h\right) }{\partial x^{i}}-\left(
\rho _{\beta }^{j}\circ h\right) \frac{\partial \left( \rho _{\alpha
}^{k}\circ h\right) }{\partial x^{j}}.%
\end{array}%
\end{equation*}
The local functions $\rho^i_\alpha$, $L^\gamma_{\alpha\beta}$ introduced in the above are called the {\textit{structure functions}} of the generalized Lie algebroid $\Big((F, \nu, N), [, ]_{F, h}, (\rho, \eta)\Big)$. Under the change of coordinates (\ref{AP1}), $\rho^i_\alpha$ and $\theta _{\alpha }^{\tilde{\imath}}$ satisfy in the following transformation rules
\[
\rho^{i'}_{\alpha'}=\Lambda _{\alpha'}^\alpha\rho^i_\alpha\frac{\partial x^{i'}}{\partial x^i},\ \ \theta^{\tilde\imath'}_{\alpha'}=\Lambda _{\alpha'}^\alpha\theta^{\tilde\imath}_{\alpha}\frac{\partial\varkappa^{\tilde\imath'}}{\partial\varkappa^{\tilde\imath}},
\]
where
$
\left\Vert \Lambda _{\alpha
{\acute{}}%
}^{\alpha }\right\Vert =\left\Vert \Lambda _{\alpha }^{\alpha
{\acute{}}%
}\right\Vert ^{-1}.
$
Also, it is known that the following relation is hold between $\rho _{\alpha }^{i}$ and $\theta _{\alpha }^{\tilde{\imath}}$
\begin{equation*}
\begin{array}{c}
\displaystyle\rho _{\alpha }^{i}\circ h\frac{\partial f\circ h}{\partial
x^{i}}=\left( \theta _{\alpha }^{\tilde{\imath}}\frac{\partial f}{\partial
\varkappa ^{\tilde{\imath}}}\right) \circ h,\ \ \ \forall f\in \mathcal{F}\left(
N\right) .%
\end{array}%
\end{equation*}
\section{The Lie algebroid generalized tangent bundle}
Let $M$ be a 4-dimensional manifold and $K$ be a 1-dimensional manifold. Considering $E=M\times K$, we introduce the Kaluza-Klein bundle $(E, \pi, M)$, where $\pi:E\longrightarrow M$ is the projection map on the first factor. Let $(x^i)$, $i\in\{1, 2, 3, 4\}$, be a coordinate system on $M$. Then we can consider a coordinate system $(x^i, y^\circ)$ on $E$, where $y^\circ$ is the fibre coordinate. Let
\[
(x^i, y^\circ)\longrightarrow (x^{i'}(x^i), {y^\circ}'(x^i, y^\circ)),
\]
be a change of coordinates on $\left( E, \pi ,M\right) $. Then the coordinate $y^\circ$ change to ${y^\circ}'$ according to the rule:
\[
{y^\circ}'=\frac{\partial {y^\circ}'}{\partial y^\circ}y^\circ.
\]
Now, let $\left( \left( F,\nu ,N\right) ,\left[ ,\right] _{F,h},\left( \rho ,\eta\right) \right)$ be a generalized Lie algebroid. Using the diagram
\begin{equation}
\begin{array}{rcl}
E &  & \left( F,\left[ ,\right] _{F,h},\left( \rho ,\eta \right) \right) \\
\pi \downarrow &  & ~\downarrow \nu \\
M & ^{\underrightarrow{~\ \ \ \ h~\ \ \ \ }} & ~\ N%
\end{array}
\end{equation}
we have the pull-back bundle $({\pi}^*(h^*F), {\pi}^*(h^*\nu), E)$. Here we consider the vector bundles morphism $\Big({\overset{\pi^{\ast }\left( h^{\ast}F\right)}{\rho }}, Id_{E}\Big)$ from $\left(\pi ^{\ast }\left( h^{\ast }F\right), \pi^{\ast }\left( h^{\ast }\nu \right) ,E\right)$ to $\left(TE,\tau _{E}, E\right)$, where
\begin{equation*}
\begin{array}{rcl}
\ \pi^{\ast }\left( h^{\ast }F\right) & ^{\underrightarrow{\overset{\pi
^{\ast }\left( h^{\ast }F\right) }{\rho }}} & TE,\\
\displaystyle Z^{\alpha }T_{\alpha }\left( u_{x}\right) & \longmapsto & %
\displaystyle\left( Z^{\alpha }\cdot \rho _{\alpha }^{i}\circ h\circ \pi
\right) \frac{\partial }{\partial x^{i}}\left( u_{x}\right),
\end{array}
\end{equation*}
and $\{T_\alpha\}_{\alpha=1}^p$ be a basis of the sections of the pull-back bundle
\[
({\pi}^*(h^*F), {\pi}^*(h^*\nu), E).
\]
Using the operation
\begin{equation*}
\begin{array}{ccc}
\Gamma \left(\pi^{\ast }\left( h^{\ast }F\right), \pi^{\ast }\left(
h^{\ast }\nu \right) ,E\right) ^{2}\!\!\!\!& ^{\underrightarrow{~\ \ \left[ ,\right]
_{\pi^{\ast }\left( h^{\ast }F\right) }~\ \ }} \!\!\!\!&\Gamma \left(\pi^{\ast
}\left( h^{\ast }F\right), \pi^{\ast }\left( h^{\ast }\nu \right)
,E\right),
\end{array}%
\end{equation*}%
defined by
\begin{equation*}
\begin{array}{ll}
\left[ T_{\alpha },T_{\beta }\right] _{\pi^{\ast }\left( h^{\ast
}F\right) }\!\!\!\!\!& =L_{\alpha \beta }^{\gamma }\circ h\circ \pi \cdot T_{\gamma },%
\vspace*{1mm} \\
\left[ T_{\alpha },fT_{\beta }\right] _{\pi^{\ast }\left( h^{\ast
}F\right) }\!\!\!\!& \displaystyle=fL_{\alpha \beta }^{\gamma }\circ h\circ \pi
T_{\gamma }+\rho _{\alpha }^{i}\circ h\circ \pi\frac{\partial f}{\partial
x^{i}}T_{\beta },\vspace*{1mm} \\
\left[ fT_{\alpha },T_{\beta }\right] _{\pi^{\ast }\left( h^{\ast
}F\right) }\!\!\!\!&=-\left[ T_{\beta },fT_{\alpha }\right] _{\pi^{\ast }\left(
h^{\ast }F\right) },%
\end{array}
\end{equation*}%
for any$f\in \mathcal{F}\left( E\right)$, it results that
\begin{equation*}
\begin{array}{c}
\left( \left(\pi^{\ast }\left( h^{\ast }F\right), \pi^{\ast }\left(
h^{\ast }\nu \right) ,E\right) ,\left[ ,\right] _{\pi^{\ast }\left(
h^{\ast }F\right) },\left( \overset{\pi^{\ast }\left( h^{\ast }F\right) }{%
\rho },Id_{E}\right) \right),
\end{array}
\end{equation*}
is a Lie algebroid which is called the {\it pull-back Lie algebroid} of the
generalized Lie algebroid $\left( \left( F,\nu ,N\right) ,\left[ ,\right]
_{F,h},\left( \rho ,\eta \right) \right)$.

If $z=z^{\alpha }t_{\alpha }\in \Gamma \left( F,\nu ,N\right)$, then we
obtain the section%
\begin{equation*}
Z=\left( z^{\alpha }\circ h\circ \pi \right) T_{\alpha }\in \Gamma \left(
\pi^{\ast }\left( h^{\ast }F\right), \pi^{\ast }\left( h^{\ast }\nu
\right) ,E\right),
\end{equation*}%
such that $Z\left( u_{x}\right) =z\left( h\left( x\right) \right) ,$ for any $%
u_{x}\in \pi ^{-1}\left( U{\cap h}^{-1}V\right) .$

Now we consider the vector bundle $\left(
\pi ^{\ast }\left( h^{\ast }F\right) \oplus TE,\overset{\oplus }{\pi}%
, E\right)$. Let $\left( \partial _{i},\dot{\partial}_\circ\right)$ be the base sections for the Lie $\mathcal{F}\left( E\right) $-algebra
\begin{equation*}
\left( \Gamma \left( TE,\tau _{E},E\right) ,+,\cdot ,\left[ ,\right]
_{TE}\right),
\end{equation*}
where $\partial_i:=\frac{\partial}{\partial x^i}$ and $\dot{\partial}_\circ:=\frac{\partial}{\partial y^\circ}$. Setting
\[
\tilde{\partial}_\alpha:=T_\alpha\oplus(\rho^i_\alpha\circ h\circ \pi)\partial_i,\ \ \ \dot{\tilde{\partial}}_\circ:=0_{\pi^*(h^* F)}\oplus\dot{\partial}_\circ,
\]
one can deduce that
\begin{equation*}
Z^{\alpha }\tilde{\partial}_{\alpha }+Y^\circ\dot{\tilde{\partial}}_\circ=Z^{\alpha }T_{\alpha }\oplus \left( Z^{\alpha }\left( \rho _{\alpha
}^{i}\circ h\circ \pi \right) \partial _{i}+Y\dot{\partial}_\circ\right),
\end{equation*}
is a section of $\left(\pi^{\ast }\left(h^{\ast }F\right) \oplus TE,\overset{%
\oplus }{\pi},E\right)$, where $Z^{\alpha }T_{\alpha }$ and $Y^\circ\dot{\partial}_\circ$ are sections of $\left(\pi^{\ast }\left( h^{\ast }F\right), \pi^{\ast }\left( h^{\ast }F\right) ,E\right)$ and $\left( VTE,\tau _{E},E\right)$, respectively. Moreover, it is easy to see that the sections $\tilde{\partial}_{1}, \ldots, \tilde{\partial}_{p}, \dot{\tilde{\partial}}_\circ$ are linearly independent. Therefore we can consider the vector subbundle $\left( \left( \rho ,\eta \right) TE,\left(
\rho ,\eta \right) \tau _{E},E\right) $ of the vector bundle\break $\left(
\pi^{\ast }\left( h^{\ast }F\right) \oplus TE,\overset{\oplus }{\pi}%
,E\right)$, for which the $\mathcal{F}\left( E\right) $-module of sections
is the $\mathcal{F}\left( E\right) $-submodule of $\left( \Gamma \left(\pi
^{\ast }\left( h^{\ast }F\right) \oplus TE,\overset{\oplus }{\pi},E\right)
,+,\cdot \right)$ generated by the set of sections $\left( \tilde{\partial}%
_{\alpha }, \dot{\tilde{\partial}}_\circ\right)$. The vector bundle $\left( \left( \rho ,\eta \right) TE,\left( \rho ,\eta
\right) \tau _{E},E\right) $ is called the \emph{generalized tangent bundle} and the base sections $\left( \tilde{\partial}_{\alpha }, \dot{\tilde{%
\partial}}_\circ\right) $ is called the \emph{natural }$\left( \rho ,\eta
\right)$-base.

The matrix of coordinate transformation on $\left( \left( \rho ,\eta \right)
TE,\left( \rho ,\eta \right) \tau _{E},E\right) $ at a change of fibred
charts is
\begin{equation*}
\left(
\begin{array}{cc}
\Lambda _{\alpha }^{\alpha
{\acute{}}%
}\circ h\circ \pi & 0\vspace*{1mm} \\
\left( \rho _{\alpha }^{i}\circ h\circ \pi \right) \displaystyle\frac{%
\partial {y^\circ}'}{\partial x^{i}} & \displaystyle\frac{\partial
{y^\circ}'}{\partial y^\circ}%
\end{array}%
\right).
\end{equation*}
\begin{theorem}
The vector bundle $\left( \left( \rho
,\eta \right) TE,\left( \rho ,\eta \right) \tau _{E},E\right)$ has a Lie algebroid structure.
\end{theorem}
\begin{proof}
We consider the vector bundles morphism $\left( \tilde{\rho},Id_{E}\right)$ from \\$\left( \left( \rho
,\eta \right) TE,\left( \rho ,\eta \right) \tau _{E},E\right)$ to $\left( TE,\tau _{E},E\right)$, where
\begin{equation*}
\begin{array}{rcl}
\left( \rho ,\eta \right) TE\!\!\! & \!\!^{\underrightarrow{\tilde{\ \ \rho
\ \ }}}\!\!\! & \!\!TE\vspace*{2mm} \\
\left( Z^{\alpha }\tilde{\partial}_{\alpha}+Y^\circ\dot{\tilde{%
\partial}}_\circ\right) \!(u_{x})\!\!\!\! & \!\!\longmapsto \!\!\! &
\!\!\left( \!Z^{\alpha }\!\left( \rho _{\alpha }^{i}{\circ }h{\circ }\pi
\!\right) \!\partial _{i}{+}Y^\circ\dot{\partial_\circ}\right) \!(u_{x}).
\end{array}
\end{equation*}
Also, we define the bracket
\begin{equation*}
\begin{array}{ccc}
\Gamma \left( \left( \rho ,\eta \right) TE,\left( \rho ,\eta \right) \tau
_{E},E\right) ^{2} & ^{\underrightarrow{~\ \ \left[ ,\right] _{\left( \rho
,\eta \right) TE}~\ \ }} & \Gamma \left( \left( \rho ,\eta \right) TE,\left(
\rho ,\eta \right) \tau _{E},E\right),
\end{array}%
\end{equation*}%
by
\begin{equation*}
\begin{array}{l}
\left[ Z_{1}^{\alpha }\tilde{\partial}_{\alpha }+Y^\circ_{1}\dot{%
\tilde{\partial}}_\circ,Z_{2}^{\beta }\tilde{\partial}_{\beta }+Y^\circ_{2}%
\dot{\tilde{\partial}}_\circ\right] _{\left( \rho ,\eta \right) TE}
\displaystyle=\left[ Z_{1}^{\alpha }T_{\alpha },Z_{2}^{\beta }T_{\beta }%
\right] _{\pi ^{\ast }\left( h^{\ast }F\right) }\\\oplus \left[ Z_{1}^{\alpha
}\left( \rho _{\alpha }^{i}\circ h\circ \pi \right) \partial _{i}+Y^\circ_{1}%
\dot{\partial}_\circ,\right.
\hfill \displaystyle\left. Z_{2}^{\beta }\left( \rho _{\beta }^{j}\circ
h\circ \pi \right) \partial _{j}+Y^\circ_{2}\dot{\partial}_\circ\right] _{TE},%
\end{array}%
\end{equation*}%
for any sections $Z_{1}^{\alpha }\tilde{\partial}_{\alpha }+Y^\circ_{1}\dot{\tilde{\partial}}_\circ$ and $Z_{2}^{\beta }\tilde{\partial}%
_{\beta }+Y^\circ_{2}\dot{\tilde{\partial}}_\circ$ of $\left( \left( \rho
,\eta \right) TE,\left( \rho ,\eta \right) \tau _{E},E\right)$. It is easy to check that $\left( \left[ ,\right] _{\left( \rho ,\eta \right)
TE},\left( \tilde{\rho},Id_{E}\right) \right) $ is a Lie algebroid
structure for the vector bundle $\left( \left( \rho ,\eta \right) TE,\left(
\rho ,\eta \right) \tau _{E},E\right)$.
\end{proof}
\begin{remark}
In particular, if $h=Id_{M},$ then the Lie
algebroid
\begin{equation*}
\begin{array}{c}
\left( \left( \left( Id_{TM},Id_{M}\right) TE,\left( Id_{TM},Id_{M}\right)
\tau _{E},E\right) ,\left[ ,\right] _{\left( Id_{TM},Id_{M}\right)
TE},\left( \widetilde{Id_{TM}},Id_{E}\right) \right),
\end{array}%
\end{equation*}%
is isomorphic with the usual Lie algebroid
\begin{equation*}
\begin{array}{c}
\left( \left( TE,\tau _{E},E\right) ,\left[ ,\right] _{TE},\left(
Id_{TE},Id_{E}\right) \right) .%
\end{array}%
\end{equation*}
This is a reason for which the Lie algebroid
\[
\begin{array}{c}
\left( \left( \left( \rho ,\eta \right) TE,\left( \rho ,\eta \right) \tau
_{E},E\right) ,\left[ ,\right] _{\left( \rho ,\eta \right) TE},\left( \tilde{%
\rho},Id_{E}\right) \right),
\end{array}
\]
 is called the Lie algebroid generalized tangent bundle.
\end{remark}
\subsection{$\left( \protect\rho ,\protect\eta \right) $-connections}
We consider the vector bundles morphism $\left( \left( \rho ,\eta
\right) \pi!, Id_{E}\right) $ given by the commutative diagram%
\begin{equation*}
\begin{array}{rcl}
\left( \rho ,\eta \right) TE & ^{\underrightarrow{~\ \left( \rho ,\eta
\right) \pi !~\ }} & \pi ^{\ast }\left( h^{\ast }F\right) \\
\left( \rho ,\eta \right) \tau _{E}\downarrow ~ &  & ~\downarrow \pi^{\ast
}\left( h^{\ast }\nu \right) \\
E~\  & ^{\underrightarrow{~Id_{E}~}} & ~\ E%
\end{array}%
\end{equation*}
defined by
\begin{equation*}
\begin{array}{c}
\left( \rho ,\eta \right) \pi!\left( \left( Z^{\alpha }\tilde{\partial}%
_{\alpha }+Y^\circ\dot{\tilde{\partial}}_\circ\right) \left(
u_{x}\right) \right) =\left( Z^{\alpha }T_{\alpha }\right) \left(
u_{x}\right) ,%
\end{array}%
\end{equation*}%
for any $Z^{\alpha }\tilde{\partial}_{\alpha }+Y^\circ\dot{\tilde{%
\partial}}_\circ\in \Gamma \left( \left( \rho ,\eta \right) TE,\left( \rho
,\eta \right) \tau _{E},E\right)$. Using this morphism, we obtain the tangent $\left( \rho
,\eta \right)$-application $\left( \left( \rho ,\eta \right) T\pi
,h\circ \pi \right) $ from \\$\left( \left( \rho ,\eta \right) TE,\left( \rho
,\eta \right) \tau _{E},E\right) $ to $\left( F,\nu ,N\right) $.
\begin{definition}
The kernel of the tangent $\left( \rho ,\eta \right)
$-application is denoted by
\begin{equation*}
\left( V\left( \rho ,\eta \right) TE,\left( \rho ,\eta \right) \tau
_{E},E\right),
\end{equation*}%
and it is called the vertical subbundle.
\end{definition}
It is remarkable that $\Big\{\dot{\tilde{\partial}}_\circ\Big\} $ is a base of the $\mathcal{F}\left( E\right) $%
-module
\[
\left( \Gamma \left( V\left( \rho ,\eta \right) TE,\left( \rho ,\eta \right)
\tau _{E},E\right) ,+,\cdot \right).
\]
The fiber bundles morphism $\left( \Pi, \pi \right) $
defined by the commutative diagram%
\begin{equation*}
\begin{array}{rcl}
V\left( \rho ,\eta \right) TE & ^{\underrightarrow{~\Pi ~\ }} & ~\ E \\
\left( \rho ,\eta \right) \tau _{E}\downarrow ~ &  & ~\downarrow \pi \\
E~\  & ^{\underrightarrow{~~~\pi ~~}} & ~\ M%
\end{array}
\end{equation*}%
such that the components of the image of the vector $Y\dot{%
\tilde{\partial}}\left( u_{x}\right) $ are the real number $Y(u_x)$ is called the \emph{%
canonical projection fiber bundle morphism.}\medskip
\begin{proposition}
The short sequence of vector bundles
\begin{equation*}
\begin{array}{ccccccccc}
0 & \hookrightarrow & V\left( \rho ,\eta \right) TE & \hookrightarrow &
\left( \rho ,\eta \right) TE & ^{\underrightarrow{~\ \left( \rho ,\eta
\right) \pi !~\ }} & \pi^{\ast }\left( h^{\ast }F\right) & ^{%
\underrightarrow{}} & 0 \\
\downarrow &  & \downarrow &  & \downarrow &  & \downarrow &  & \downarrow
\\
E & ^{\underrightarrow{~Id_{E}~}} & E & ^{\underrightarrow{~Id_{E}~}} & E &
^{\underrightarrow{~Id_{E}~}} & E & ^{\underrightarrow{~Id_{E}~}} & E%
\end{array}
\end{equation*}
is exact.
\end{proposition}
\begin{definition}
A manifolds morphism $\left( \rho
,\eta \right) \Gamma $ from $\left( \rho ,\eta \right) TE$ to $V\left(
\rho ,\eta \right) TE$ defined by
\begin{equation*}
\begin{array}{c}
\left( \rho ,\eta \right) \Gamma \left( Z^{\gamma }\tilde{\partial}_{\gamma
}+Y^\circ\dot{\tilde{\partial}}_\circ\right) \left( u_{x}\right)
=\left( Y^\circ+\left( \rho ,\eta \right) \Gamma^\circ_{\gamma }Z^{\gamma
}\right) \dot{\tilde{\partial}}_\circ\left( u_{x}\right) ,%
\end{array}
\end{equation*}%
such that the vector bundles morphism $\left( \left( \rho ,\eta
\right) \Gamma ,Id_{E}\right) $ is a split to the left in the above exact
sequence, is called $\left( \rho ,\eta \right) $-connection for
the Kaluza-Klein bundle $\left( E, \pi ,M\right) $.
\end{definition}
In particular case, the $\left( \rho ,Id_{M}\right) $-connection is called $\rho $-connection and it is denoted by $\rho \Gamma $ and in the classical case, the $\left(
Id_{TM},Id_{M}\right) $-connection is called connection and it is denoted by $\Gamma $.
\begin{theorem}
If $\left( \rho ,\eta \right) \Gamma $ is a $\left( \rho ,\eta \right) $-connection for the Kaluza-Klein bundle \\$%
\left( E, \pi ,M\right) ,$\emph{\ then its components satisfy the law of
transformation }%
\begin{equation}\label{AP2}
\begin{array}{c}
\left( \rho ,\eta \right) \Gamma _{\gamma'}^{\circ'}=\frac{\partial {y^\circ}'}{\partial y^\circ}\left[ \rho _{\gamma }^{k}\circ h\circ \pi \frac{\partial
y^\circ}{\partial x^{k}}+\left( \rho ,\eta \right) \Gamma^\circ_{\gamma }\right]
\Lambda _{\gamma'}^{\gamma }\circ h\circ \pi.
\end{array}
\end{equation}
\end{theorem}
In the particular case of Lie algebroids, $\left( \eta ,h\right)
=\left( Id_{M},Id_{M}\right)$, the above relation reduce to
\begin{equation*}
\begin{array}{c}
\rho \Gamma^{\circ'} _{\gamma'}=\frac{\partial {y^\circ}'}{\partial y^\circ}\left[ \rho _{\gamma }^{k}\circ \pi \frac{\partial y^\circ}{\partial x^{k}}+\rho\Gamma^\circ _{\gamma }\right] \Lambda _{\gamma'}^{\gamma }\circ \pi,
\end{array}
\end{equation*}
and in the classical case, $\left( \rho ,\eta ,h\right) =\left(
Id_{TM},Id_{M},Id_{M}\right)$, the relations (\ref{AP2}) reduce to
\begin{equation*}
\begin{array}{c}
\Gamma_{k'}^{\circ'}=\frac{\partial {y^\circ}'}{\partial y^\circ}\left[ \frac{\partial y^\circ}{\partial x^{k}}+\Gamma^\circ
_{k}\right] \frac{\partial x^{k}}{\partial x^{k'}}\circ \tau _{M}.%
\end{array}
\end{equation*}
\begin{definition}
Let $\left( \rho ,\eta \right) \Gamma $ be a $\left(
\rho ,\eta \right)$-connection for the Kaluza-Klein bundle $\left( E, \pi ,M\right)$. Then the kernel of the vector bundles morphism $\left( \left(
\rho ,\eta \right) \Gamma ,Id_{E}\right)$ is a vector subbundle of vector bundle $\left( \left( \rho
,\eta \right) TE,\left( \rho ,\eta \right) \tau _{E},E\right) $ which is called the horizontal subbundle and it is denoted by $\left( H\left( \rho ,\eta \right) TE,\left( \rho ,\eta \right) \tau _{E},E\right) $.
\end{definition}
We put
\begin{equation*}
\begin{array}[t]{l}
\tilde{\delta}
_{\alpha }=\tilde{\partial}_{\alpha }-\left( \rho ,\eta \right) \Gamma
_{\alpha }\dot{\tilde{\partial}}_\circ=T_{\alpha }\oplus
\left( \left( \rho _{\alpha }^{i}\circ h\circ \pi \right) \partial
_{i}-\left( \rho ,\eta \right) \Gamma^\circ _{\alpha }\dot{\partial}%
_\circ\right).
\end{array}
\end{equation*}%
 It is easy to see that $\{\tilde{\delta}_\alpha\}$ is a basis for the $\mathcal{F}\left( E\right) $%
-module
\begin{equation*}
\left( \Gamma \left( H\left( \rho ,\eta \right) TE,\left( \rho ,\eta \right)
\tau _{E},E\right) ,+,\cdot \right),
\end{equation*}
and $\left( \tilde{\delta}_{\alpha }, \dot{\tilde{\partial}}_\circ\right)$ is a basis for the $\mathcal{F}\left( E\right)$
-module
\begin{equation*}
\left( \Gamma \left(\left( \rho ,\eta \right) TE,\left( \rho ,\eta \right)
\tau _{E}, E\right) ,+,\cdot \right),
\end{equation*}
which is called the \emph{adapted }$\left( \rho ,\eta \right) $%
\emph{-base}. Moreover, the following equality holds
\begin{equation*}
\begin{array}{l}
\Gamma \left( \tilde{\rho},Id_{E}\right) \left( \tilde{\delta}_{\alpha
}\right) =\left( \rho _{\alpha }^{i}\circ h\circ \pi \right) \partial
_{i}-\left( \rho ,\eta \right) \Gamma^\circ_{\alpha }\dot{\partial}_\circ,
\end{array}
\end{equation*}
where $\left( \delta _{i},\dot{\partial}_\circ\right) $ is the adapted base
for the $\mathcal{F}\left( E\right) $-module $\left( \Gamma \left( TE,\tau
_{E},E\right) ,+,\cdot \right)$.
\begin{theorem}
The following equalities hold
\begin{align*}
\left[ \tilde{\delta}_{\alpha },\tilde{\delta}_{\beta }\right] _{\left( \rho
,\eta \right) TE}&=L_{\alpha \beta }^{\gamma }\circ \left( h\circ \pi \right)
\tilde{\delta}_{\gamma }+\left( \rho ,\eta ,h\right) \mathbb{R}^\circ_{\ \alpha
\beta }\dot{\tilde{\partial}}_\circ,\\
\left[ \tilde{\delta}_{\alpha }, \dot{\tilde{\partial}}_\circ\right]
_{\left(\rho, \eta \right) TE}&=\Gamma\left( \tilde{\rho},Id_{E}\right)
\left(\dot{\tilde{\partial}}_\circ\right) \left( \left( \rho
,\eta \right) \Gamma^\circ_{\alpha }\right)\dot{\tilde{\partial}}_\circ,%
\end{align*}
where
\begin{align*}
\left( \rho, \eta, h\right) \mathbb{R}^\circ_{\ \alpha \beta }&=\Gamma \left(
\tilde{\rho},Id_{E}\right) \left( \tilde{\delta}_{\beta }\right) \left(
\left( \rho ,\eta \right) \Gamma^\circ_{\alpha }\right)
-\Gamma \left( \tilde{\rho},Id_{E}\right) \left( \tilde{\delta}%
_{\alpha }\right) \left( \left( \rho ,\eta \right) \Gamma^\circ_{\beta
}\right)\\
&\ \ \ +\left( L_{\alpha \beta }^{\gamma }\circ h\circ \pi \right)
\left( \rho ,\eta \right) \Gamma^\circ_{\gamma }.
\end{align*}
\end{theorem}
If $\left( d\tilde{z}^{\alpha }, d\tilde{y}^\circ\right)$ is the natural dual $\left(
\rho ,\eta \right) $-base of the natural $\left( \rho ,\eta \right)$-base $
\left(\tilde{\partial}_{\alpha }, \dot{\tilde{\partial}}_\circ\right)$, then it is easy to check that $( d\tilde{z}^{\alpha }, \delta \tilde{y}^\circ)$ is the dual $\left(
\rho ,\eta \right) $-base of the adapted $\left( \rho ,\eta \right)$-base $\left(\tilde{\delta}_{\alpha }, \dot{\tilde{\partial}}_\circ\right)$, where
\begin{equation*}
\begin{array}{c}
\delta{\tilde{y}}^\circ=\left( \rho ,\eta \right)\Gamma^\circ_{\alpha }d\tilde{z%
}^{\alpha }+d{\tilde{y}}^\circ.
\end{array}
\end{equation*}
The base $\left( d\tilde{z}^{\alpha },\delta{\tilde{y}}^\circ\right) $
is called the \emph{adapted dual }$\left( \rho ,\eta \right) $\emph{-base.}
\subsection{Distinguished linear $\left( \protect\rho ,\protect\eta \right) $%
-connections}
Let $\left( \rho ,\eta
\right) \Gamma $ be a $\left( \rho ,\eta \right) $-connection for the Kaluza-Klein
bundle $\left( E, \pi ,M\right)$ and let
\begin{equation*}
\left( \mathcal{T}~_{q, 1}^{p, 1}\left( \left( \rho ,\eta \right) TE,\left(
\rho ,\eta \right) \tau _{E},E\right) ,+,\cdot \right),
\end{equation*}%
be the $\mathcal{F}\left( E\right) $-module of $\left(
_{q, 1}^{p, 1}\right) $-tensor fields of the generalized tangent bundle
\begin{equation*}
\left( H\left( \rho ,\eta \right) TE,\left( \rho ,\eta \right) \tau
_{E},E\right) \oplus \left( V\left( \rho ,\eta \right) TE,\left( \rho ,\eta
\right) \tau _{E},E\right) .
\end{equation*}

An arbitrarily tensor field $T$ with respect to the adapted $\left( \rho ,\eta \right)$-base $\left(\tilde{\delta}_{\alpha }, \dot{\tilde{\partial}}_\circ\right)$ is written as
\begin{equation*}
\begin{array}{c}
T=T_{\beta _{1}...\beta _{q}\circ}^{\alpha _{1}...\alpha
_{p}\circ}\tilde{\delta}_{\alpha _{1}}\otimes ...\otimes \tilde{%
\delta}_{\alpha _{p}}\otimes d\tilde{z}^{\beta _{1}}\otimes ...\otimes d%
\tilde{z}^{\beta _{q}}\otimes\dot{\tilde{\partial}}_\circ\otimes\delta\tilde{y}^{\circ}.%
\end{array}%
\end{equation*}
Let
$
\left( \mathcal{T}~\left( \left( \rho ,\eta \right) TE,\left( \rho ,\eta
\right) \tau _{E},E\right) ,+,\cdot ,\otimes \right)
$
be the tensor fields algebra of generalized tangent bundle $\left( \left(
\rho ,\eta \right) TE,\left( \rho ,\eta \right) \tau _{E},E\right) $ and let
\begin{equation*}
\begin{array}{l}
\left( X,T\right) ^{\ \underrightarrow{\left( \rho ,\eta \right) D}\,}%
\vspace*{1mm}\left( \rho ,\eta \right) D_{X}T,
\end{array}%
\end{equation*}%
be a covariant $\left( \rho ,\eta \right) $-derivative for the tensor
algebra
\begin{equation*}
\left( \mathcal{T}~\left( \left( \rho ,\eta \right) TE,\left( \rho ,\eta
\right) \tau _{E},E\right) ,+,\cdot ,\otimes \right),
\end{equation*}%
of the generalized tangent bundle
$
\left( \left( \rho ,\eta \right) TE,\left( \rho ,\eta \right) \tau
_{E},E\right)
$
which preserves the horizontal and vertical interior differential systems by
parallelism. The real local functions
\begin{equation*}
\left( \left( \rho ,\eta \right) H_{\beta \gamma }^{\alpha }, \left( \rho
,\eta \right) H^\circ_{\circ\gamma},\left( \rho ,\eta \right) V_{\beta\circ}^{\alpha
},\left( \rho ,\eta \right) V^\circ_{\circ\circ}\right),
\end{equation*}%
defined by the following equalities:
\begin{equation*}
\begin{array}{ll}
\left( \rho ,\eta \right) D_{\tilde{\delta}_{\gamma }}\tilde{\delta}_{\beta
}=\left( \rho ,\eta \right) H_{\beta \gamma }^{\alpha }\tilde{\delta}%
_{\alpha }, & \left( \rho ,\eta \right) D_{\tilde{\delta}_{\gamma }}\dot{\tilde{\partial}}_\circ=\left( \rho ,\eta \right) H^\circ_{\circ\gamma }%
\dot{\tilde{\partial}}_\circ,\\
\left( \rho ,\eta \right) D_{\dot{\tilde{\partial}}_\circ}\tilde{%
\delta}_{\beta }=\left( \rho ,\eta \right) V_{\beta\circ}^{\alpha }\tilde{\delta%
}_{\alpha }, & \left( \rho ,\eta \right) D_{\dot{\tilde{\partial}%
}_\circ}\dot{\tilde{\partial}}_\circ=\left( \rho ,\eta \right)
V^\circ_{\circ\circ}\dot{\tilde{\partial}}_\circ,
\end{array}
\end{equation*}
are the components of a linear $\left( \rho ,\eta \right) $-connection $%
\left( \left( \rho ,\eta \right) H,\left( \rho ,\eta \right) V\right) $ for
the generalized tangent bundle $\left( \left( \rho ,\eta \right) TE,\left(
\rho ,\eta \right) \tau _{E},E\right) $ which is called the \emph{%
distinguished linear }$\left( \rho ,\eta \right) $\emph{-connection.}

In the particular case of Lie algebroids, $h=Id_{M}=\eta ,$ we obtain the
\emph{distinguished linear }$\rho $\emph{-connection.} The components of a
distinguished linear $\rho $-connection $\left( \rho H,\rho V\right) $ are denoted by
\begin{equation*}
\left( \rho H_{\beta \gamma }^{\alpha },\rho H^\circ_{\circ\gamma },\rho V_{\beta\circ}^{\alpha },\rho V^\circ_{\circ\circ}\right) .
\end{equation*}
In addition, if $\rho =Id_{TM},$ then we obtain the classical \emph{%
distinguished linear connection. }The components of a distinguished linear
connection $\left( H,V\right) $ are denoted by
\begin{equation*}
\left( H_{jk}^{i},H^\circ_{\circ k},V_{j\circ}^{i},V^\circ_{\circ\circ}\right) .
\end{equation*}
\begin{theorem}
If  $((\rho ,\eta )H,(\rho ,\eta )V)$ \emph{is a
distinguished linear} $(\rho ,\eta )$-\emph{connection for the generalized
tangent bundle }$\left( \left( \rho ,\eta \right) TE,\left( \rho ,\eta
\right) \tau _{E},E\right) $\emph{, then its components satisfy the change
relations: }

\begin{equation*}
\begin{array}{ll}
\left( \rho ,\eta \right) H_{\beta'\gamma'}^{\alpha'}\!\! & =\Lambda _{\alpha }^{\alpha'}
\circ h\circ \pi \cdot \left[ \Gamma \left( \tilde{\rho},Id_{E}\right)
\left( \tilde{\delta}_{\gamma }\right) \left( \Lambda _{\beta}^{\alpha }\circ h\circ \pi \right) +\right. \vspace*{1mm} \\
& +\left. \left( \rho ,\eta \right) H_{\beta \gamma }^{\alpha }\cdot \Lambda_{\beta'}^{\beta }\circ h\circ \pi \right] \cdot \Lambda _{\gamma'}^{\gamma }\circ h\circ \pi ,\vspace*{2mm} \\
\left( \rho ,\eta \right) H^{\circ'}_{\circ'\gamma'}\!\! & =\frac{\partial {y^\circ}'}{\partial y^\circ}\cdot \left[ \Gamma \left(
\tilde{\rho},Id_{E}\right) \left(\tilde{\delta}_{\gamma }\right) \left(
\frac{\partial y^\circ}{\partial {y^\circ}'}\right) +\right. \vspace*{1mm} \\
& \left. +\left( \rho ,\eta \right) H^\circ_{\circ\gamma }\cdot \frac{\partial y^\circ}{%
\partial {y^\circ}'}\right] \cdot \Lambda _{\gamma'}^{\gamma }\circ h\circ \pi ,\vspace*{2mm} \\
\left( \rho ,\eta \right) V_{\beta'\circ'}^{\alpha'}\!\! & =\Lambda _{\alpha }^{\alpha ^{\prime }}\circ h\circ \pi \cdot \left(
\rho ,\eta \right) V_{\beta\circ}^{\alpha }\cdot \Lambda _{\beta'}^{\beta }\circ h\circ \pi \cdot \frac{\partial y^\circ}{\partial {y^\circ}'},%
\vspace*{2mm} \\
\left( \rho ,\eta \right) V^{\circ'}_{\circ'\circ'}\!\! & =\frac{\partial {y^\circ}'}{\partial y^\circ}\cdot \left( \rho ,\eta
\right) V^\circ_{\circ\circ}\cdot \frac{\partial y^\circ}{\partial {y^\circ}'}\cdot \frac{%
\partial y^\circ}{\partial {y^\circ}'}.%
\end{array}%
\end{equation*}
\end{theorem}
In the particular case of Lie algebroids, $\left( \eta ,h\right) =\left( Id_{M},Id_{M}\right) ,$ we obtain
\begin{equation*}
\begin{array}{ll}
\rho H_{\beta'\gamma'}^{\alpha'}\!\!&=\Lambda _{\alpha }^{\alpha'}\circ \pi \cdot \left[ \Gamma \left( \tilde{\rho},Id_{E}\right) \left(
\tilde{\delta}_{\gamma }\right) \left( \Lambda _{\beta'}^{\alpha }\circ \pi \right)+\rho H_{\beta \gamma }^{\alpha }\cdot \Lambda
_{\beta'}^{\beta }\circ \pi \right] \cdot \Lambda _{\gamma'}^{\gamma }\circ \pi,\\
\rho H^{\circ'}_{\circ'\gamma'}\!\! & =\frac{\partial {y^\circ}'}{\partial y^\circ}\cdot \left[ \Gamma
\left( \tilde{\rho}, Id_{E}\right) \left( \tilde{\delta}_{\gamma }\right)
\left( \frac{\partial y^\circ}{\partial {y^\circ}'}\right) +\rho
H^\circ_{\circ\gamma }\cdot \frac{\partial y^\circ}{\partial {y^\circ}'}\right]
\cdot \Lambda _{\gamma'}^{\gamma }\circ \pi,\\
\rho V_{\beta'\circ'}^{\alpha'}\!\! &=\Lambda _{\alpha'}^{\alpha }\circ \pi \cdot \rho V_{\beta\circ}^{\alpha }\cdot \Lambda _{\beta'}^{\beta }\circ \pi \cdot \frac{\partial y^\circ}{\partial {y^\circ}'},%
\vspace*{2mm} \\
\rho V^{\circ'}_{\circ'\circ'}\!\! & =\frac{\partial{y^\circ}'}{\partial y^\circ}\cdot \rho
V^\circ_{\circ\circ}\cdot \frac{\partial y^\circ}{\partial{y^\circ}'}\cdot \frac{%
\partial y^\circ}{\partial{y^\circ}'}.%
\end{array}
\end{equation*}
Also, in the classical case, $\left( \rho ,\eta ,h\right) =\left(
Id_{TM},Id_{M},Id_{M}\right) ,$ we deduce that the components of a
distinguished linear connection $\left( H,V\right) $ satisfy the
change relations:
\begin{equation*}
\begin{array}{cl}
H_{j'k'}^{i'}& =\frac{\partial x^{i'}}{\partial x^{i}}\circ \pi \cdot \left[ \frac{\delta }{\delta x^{k}}\left(
\frac{\partial x^{i}}{\partial x^{j'}}\circ \pi \right)+H_{jk}^{i}\cdot \frac{\partial x^{j}}{\partial x^{j'}}\circ \pi \right] \cdot \frac{\partial x^{k}}{\partial x^{k'}}\circ \pi ,\vspace*{2mm}\\
H^{\circ'}_{\circ'k'}&=\frac{\partial {y^\circ}'}{\partial y^\circ}\cdot \left[ \frac{%
\delta }{\delta x^{k}}\left( \frac{\partial y^\circ}{\partial {y^\circ}'}%
\right)+H^\circ_{\circ k}\cdot \frac{\partial y^\circ}{\partial{y^\circ}'}%
\right] \cdot \frac{\partial x^{k}}{\partial x^{k'}}\circ \pi ,\vspace*{2mm}\\
V_{j'\circ'}^{i'}& =\frac{\partial x^{i'}}{\partial x^{i}}\circ \pi \cdot V_{j\circ}^{i}\cdot \frac{\partial x^{j}}{
\partial x^{j'}}\circ \pi \cdot \frac{\partial y^\circ}{\partial {y^\circ}'},\vspace*{%
3mm}\\
V^{\circ'}_{\circ'\circ'}\!\!\!&=\frac{\partial {y^\circ}'}{\partial y^\circ}\cdot V^\circ_{\circ\circ}\cdot
\frac{\partial y^\circ}{\partial{y^\circ}'}\cdot \frac{\partial y^\circ}{%
\partial{y^\circ}'}.%
\end{array}%
\end{equation*}
\begin{example}
The local real functions
\begin{equation*}
\begin{array}[b]{c}
\left( \frac{\partial \left( \rho ,\eta \right) \Gamma^\circ_{\gamma }}{%
\partial y^\circ},\frac{\partial \left( \rho ,\eta \right) \Gamma^\circ_{\gamma
}}{\partial y^\circ},0,0\right),
\end{array}%
\end{equation*}
are the components of a distinguished linear $\left( \rho ,\eta \right) $-connection for the generalized tangent bundle $\left( \left( \rho
,\eta \right) TE,\left( \rho ,\eta \right) \tau _{E},E\right) ,$ which is called the Berwald linear $\left( \rho ,\eta \right) $-connection.
\end{example}
Note that the Berwald linear $(Id_{TM},Id_{M})$-connection is the usual \emph{%
Berwald linear connection.}
\begin{theorem}
If the generalized tangent bundle $\!(\!(\rho
,\!\eta )T\!E,\!(\rho ,\!\eta )\tau _{E},\!E\!)$ is endowed with a
distinguished linear $\!(\rho ,\!\eta )$-connection $((\rho ,\eta
)H,(\rho ,\eta )V)$, then for any
\begin{equation*}
\begin{array}[b]{c}
X=Z^{\alpha }\tilde{\delta}_{\alpha }+Y^\circ\dot{\tilde{\partial}}_\circ\in \Gamma (\!(\rho ,\eta )TE,\!(\rho ,\!\eta )\tau _{E},\!E),
\end{array}%
\end{equation*}%
and
\begin{equation*}
T\in \mathcal{T}_{q1}^{p1}\!(\!(\rho ,\eta )TE,\!(\rho ,\eta )\tau _{E},\!E),
\end{equation*}%
we obtain the formula:
\begin{equation*}
\begin{array}{l}
\left( \rho ,\eta \right) D_{X}\left( T_{\beta _{1}\cdots\beta
_{q}\circ}^{\alpha _{1}\cdots\alpha _{p}\circ}\tilde{\delta}%
_{\alpha _{1}}\otimes \cdots\otimes \tilde{\delta}_{\alpha _{p}}\otimes d\tilde{%
z}^{\beta _{1}}\otimes \cdots\otimes d\tilde{z}^{\beta _{q}}\otimes\dot{%
\tilde{\partial}}_\circ\otimes \delta \tilde{y}^{\circ}\right)\\
\hspace*{9mm}=Z^{\gamma }T_{\beta _{1}\ldots\beta _{q}\circ\mid \gamma
}^{\alpha _{1}\ldots\alpha _{p}\circ}\tilde{\delta}_{\alpha
_{1}}\otimes \cdots\otimes \tilde{\delta}_{\alpha _{p}}\otimes d\tilde{z}%
^{\beta _{1}}\otimes \cdots\otimes d\tilde{z}^{\beta _{q}}\otimes\dot{\tilde{\partial}}_\circ\otimes\delta \tilde{y}^\circ\\
\hspace*{11mm}+Y^{\circ}T_{\beta _{1}\ldots\beta _{q}\circ}^{\alpha _{1}\ldots\alpha
_{p}\circ}\mid _{\circ}\tilde{\delta}_{\alpha _{1}}\otimes \cdots\otimes\tilde{\delta}_{\alpha _{p}}\otimes d\tilde{z}^{\beta
_{1}}\otimes \cdots\otimes d\tilde{z}^{\beta _{q}}\otimes\dot{\tilde{\partial}}_{\circ}\otimes \delta \tilde{y}^{\circ},
\end{array}
\end{equation*}%
where
\begin{equation*}
\begin{array}{l}
T_{\beta _{1}\cdots\beta _{q}\circ\mid \gamma }^{\alpha _{1}\cdots\alpha
_{p}\circ}=\vspace*{2mm}\Gamma \left( \tilde{\rho},Id_{E}\right)
\left( \tilde{\delta}_{\gamma }\right) T_{\beta _{1}\ldots\beta
_{q}\circ}^{\alpha _{1}\cdots\alpha _{p}\circ} \\
\hspace*{8mm}+\left( \rho ,\eta \right) H_{\alpha \gamma }^{\alpha
_{1}}T_{\beta _{1}\cdots\beta _{q}\circ}^{\alpha \alpha _{2}\cdots\alpha
_{p}\circ}+\cdots+\vspace*{2mm}\left( \rho ,\eta \right) H_{\alpha
\gamma }^{\alpha _{p}}T_{\beta _{1}\cdots\beta _{q}\circ}^{\alpha
_{1}\cdots\alpha _{p-1}\circ} \\
\hspace*{8mm}-\left( \rho ,\eta \right) H_{\beta _{1}\gamma }^{\beta
}T_{\beta \beta _{2}\cdots\beta _{q}\circ}^{\alpha _{1}\cdots\alpha
_{p}\circ}-\cdots-\vspace*{2mm}\left( \rho ,\eta \right) H_{\beta
_{q}\gamma }^{\beta }T_{\beta _{1}\cdots\beta _{q-1}\beta
\circ}^{\alpha _{1}\cdots\alpha _{p}\circ},
\end{array}
\end{equation*}%
and
\begin{equation*}
\begin{array}{l}
T_{\beta _{1}\cdots\beta _{q}\circ}^{\alpha _{1}\cdots\alpha
_{p}\circ}\mid _{\circ}=\Gamma \left( \tilde{\rho},Id_{E}\right) \left(
\dot{\tilde{\partial}}_{\circ}\right) T_{\beta _{1}\cdots\beta
_{q}\circ}^{\alpha _{1}\cdots\alpha _{p}\circ} \\
\hspace*{8mm}+\left( \rho ,\eta \right) V_{\alpha \circ}^{\alpha _{1}}T_{\beta
_{1}\cdots\beta _{q}\circ}^{\alpha \alpha _{2}\cdots\alpha
_{p}\circ}+\cdots+\left( \rho ,\eta \right) V_{\alpha \circ}^{\alpha
_{p}}T_{\beta _{1}\cdots\beta _{q}\circ}^{\alpha _{1}\cdots\alpha
_{p-1}\alpha \circ}\vspace*{2mm} \\
\hspace*{8mm}-\left( \rho ,\eta \right) V_{\beta _{1}\circ}^{\beta }T_{\beta
\beta _{2}\cdots\beta _{q}\circ}^{\alpha _{1}\cdots\alpha
_{p}\circ}-\cdots-\left( \rho ,\eta \right) V_{\beta _{q}\circ}^{\beta
}T_{\beta _{1}\cdots\beta _{q-1}\beta \circ}^{\alpha _{1}\cdots\alpha
_{p}\circ}\vspace*{2mm}.
\end{array}
\end{equation*}
\end{theorem}
In the particular case of Lie algebroids, $\left( \eta ,h\right) =\left( Id_{M},Id_{M}\right) ,$ we obtain
\begin{equation*}
\begin{array}{l}
T_{\beta _{1}\cdots\beta _{q}\circ\mid \gamma }^{\alpha _{1}\cdots\alpha
_{p}\circ}=\vspace*{2mm}\Gamma \left( \tilde{\rho},Id_{E}\right)
\left( \tilde{\delta}_{\gamma }\right) T_{\beta _{1}\ldots\beta
_{q}\circ}^{\alpha _{1}\cdots\alpha _{p}\circ} \\
\hspace*{8mm}+\rho H_{\alpha \gamma }^{\alpha
_{1}}T_{\beta _{1}\cdots\beta _{q}\circ}^{\alpha \alpha _{2}\cdots\alpha
_{p}\circ}+\cdots+\vspace*{2mm}\rho H_{\alpha
\gamma }^{\alpha _{p}}T_{\beta _{1}\cdots\beta _{q}\circ}^{\alpha
_{1}\cdots\alpha _{p-1}\circ} \\
\hspace*{8mm}-\rho H_{\beta _{1}\gamma }^{\beta
}T_{\beta \beta _{2}\cdots\beta _{q}\circ}^{\alpha _{1}\cdots\alpha
_{p}\circ}-\cdots-\vspace*{2mm}\rho H_{\beta
_{q}\gamma }^{\beta }T_{\beta _{1}\cdots\beta _{q-1}\beta
\circ}^{\alpha _{1}\cdots\alpha _{p}\circ},
\end{array}
\end{equation*}%
and
\begin{equation*}
\begin{array}{l}
T_{\beta _{1}\cdots\beta _{q}\circ}^{\alpha _{1}\cdots\alpha
_{p}\circ}\mid _{\circ}=\Gamma \left( \tilde{\rho},Id_{E}\right) \left(
\dot{\tilde{\partial}}_{\circ}\right) T_{\beta _{1}\cdots\beta
_{q}\circ}^{\alpha _{1}\cdots\alpha _{p}\circ} \\
\hspace*{8mm}+\rho V_{\alpha \circ}^{\alpha _{1}}T_{\beta
_{1}\cdots\beta _{q}\circ}^{\alpha \alpha _{2}\cdots\alpha
_{p}\circ}+\cdots+\rho V_{\alpha \circ}^{\alpha
_{p}}T_{\beta _{1}\cdots\beta _{q}\circ}^{\alpha _{1}\cdots\alpha
_{p-1}\alpha \circ}\vspace*{2mm} \\
\hspace*{8mm}-\rho V_{\beta _{1}\circ}^{\beta }T_{\beta
\beta _{2}\cdots\beta _{q}\circ}^{\alpha _{1}\cdots\alpha
_{p}\circ}-\cdots-\rho V_{\beta _{q}\circ}^{\beta
}T_{\beta _{1}\cdots\beta _{q-1}\beta \circ}^{\alpha _{1}\cdots\alpha
_{p}\circ}\vspace*{2mm}.
\end{array}
\end{equation*}
Moreover, in the classical case, $\left( \rho ,\eta ,h\right) =\left(
Id_{TM},Id_{M},Id_{M}\right)$, we obtain
\begin{equation*}
\begin{array}{l}
T_{j_{1}\cdots j_{q}\circ\mid k}^{i_{1}\cdots i_{p}\circ}=\vspace*{%
2mm}\delta _{k}\left(
T_{j_{1}\cdots j_{q}\circ}^{i_{1}\cdots i_{p}\circ}\right) \\
\hspace*{8mm}%
+H_{ik}^{i_{1}}T_{j_{1}\cdots j_{q}\circ}^{ii_{2}\cdots i_{p}\circ}+\cdots+%
\vspace*{2mm}H_{ik}^{i_{p}}T_{\beta _{1}\cdots\beta
_{q}\circ}^{i_{1}\cdots i_{p-1}i\circ} \\
\hspace*{8mm}%
-H_{j_{1}k}^{j}T_{jj_{2}\cdots j_{q}\circ}^{i_{1}\cdots i_{p}\circ}-\cdots-%
\vspace*{2mm}H_{j_{q}k}^{j}T_{j_{1}\cdots j_{q-1}j\circ}^{\alpha
_{1}\cdots\alpha _{p}\circ},
\end{array}%
\end{equation*}%
and
\begin{equation*}
\begin{array}{l}
T_{j_{1}\cdots j_{q}\circ}^{i_{1}\cdots i_{p}\circ}\mid _{\circ}=\dot{%
\partial}_{\circ}\left( T_{\beta _{1}\cdots\beta _{q}\circ}^{\alpha
_{1}\cdots\alpha _{p}\circ}\right) \vspace*{2mm} \\
\hspace*{8mm}%
+V_{i\circ}^{i_{1}}T_{j_{1}\cdots j_{q}\circ}^{ii_{2}\cdots i_{p}\circ}+\cdots+V_{i\circ}^{i_{p}}T_{\beta _{1}\cdots\beta _{q}\circ}^{i_{1}\cdots i_{p-1}i\circ}%
\vspace*{2mm} \\
\hspace*{8mm}%
-V_{j_{1}\circ}^{j}T_{jj_{2}\cdots j_{q}\circ}^{i_{1}\cdots i_{p}\circ}-\cdots-V_{j_{q}\circ}^{j}T_{j_{1}\cdots j_{q-1}j\circ}^{i_{1}\cdots i_{p}\circ}%
\vspace*{2mm}.
\end{array}
\end{equation*}
\section{The $\left( g,h\right) $-lift of a differentiable curve}
In this section we use the Kaluza-Klein bundle such that the fibre $K$ is an 1-dimensional real vector space. Thus we obtain the Kaluza-Klein vector bundle $(E, \pi, M)$.

Let $c:I\rightarrow M$ be a differentiable curve. Then
\begin{equation*}
\begin{array}{c}
\left( E_{|\func{Im}\left( \eta \circ h\circ c\right) }, {\pi}_{|\func{Im}%
\left( \eta \circ h\circ c\right) },\func{Im}\left( \eta \circ h\circ
c\right) \right),
\end{array}%
\end{equation*}%
is a vector subbundle of $\left( E, \pi ,M\right)$.
\begin{definition}
Let
\begin{equation*}
\begin{array}{ccc}
I & ^{\underrightarrow{\ \ \dot{c}\ \ }} & E_{|\func{Im}\left( \eta \circ
h\circ c\right) },\\
t & \longmapsto & y^\circ\left( t\right) s_\circ\left( \eta \circ h\circ c\left(
t\right) \right),
\end{array}
\end{equation*}%
be a differentiable curve. If there exists a manifolds morphism $g: E\rightarrow F$ such that the following conditions are hold:

1) $(g, h)$ is a vector bundle morphism from $(E, \pi, M)$ to  $( F, \nu , N)$,

2) $\rho \circ g\circ \dot{c}\left( t\right) =\displaystyle\frac{%
d\left( \eta \circ h\circ c\right) ^{i}\left( t\right) }{dt}\frac{\partial }{%
\partial x^{i}}\left( \left( \eta \circ h\circ c\right) \left( t\right)
\right) ,$ for any $t\in I$,\\
then $\dot{c}$ is called the $\left( g,h\right) $-lift of the differentiable curve $c$.
\end{definition}
\begin{remark}
The second condition is equivalent with the
following:
\begin{equation*}
\begin{array}[b]{c}
\rho _{\alpha }^{i}\left( \eta \circ h\circ c\left( t\right) \right) \cdot
g^{\alpha }_\circ\left( h\circ c\left( t\right) \right) \cdot y^\circ\left(
t\right) =\frac{d\left( \eta \circ h\circ c\right) ^{i}\left( t\right) }{dt}%
,~i\in 1, \ldots, 4.
\end{array}%
\end{equation*}
\end{remark}
\begin{definition}
If $%
\begin{array}{ccc}
\dot{c}: I \longrightarrow E_{|\func{Im}\left( \eta \circ h\circ
c\right) },
\end{array}%
$ is a differentiable $\left( g,h\right) $-lift of the differentiable curve $%
c,$ then the section%
\begin{equation*}
\begin{array}{ccc}
\func{Im}\left( \eta \circ h\circ c\right) & ^{\underrightarrow{u\left( c,%
\dot{c}\right) }} & E_{|\func{Im}\left( \eta \circ h\circ c\right) }\vspace*{%
1mm},\\
\eta \circ h\circ c\left( t\right) & \longmapsto & \dot{c}\left( t\right),
\end{array}
\end{equation*}%
is called the canonical section associated to $\left( c,\dot{c}\right)$.
\end{definition}
\begin{definition}
If $\left( g,h\right)$ has the
components
\begin{equation*}
\begin{array}{c}
g^{\alpha}_\circ,\ \ \alpha \in 1, \ldots, p,
\end{array}%
\end{equation*}%
such that for any local vector chart $\left(
V,t_{V}\right) $ of $\left( F,\nu ,N\right)$, there exists the real
functions
\begin{equation*}
\begin{array}{ccc}
V & ^{\underrightarrow{~\ \ \ \tilde{g}^\circ_{\alpha }~\ \ }}\!\!\!\!& \mathbb{R},\ \ \alpha \in 1, \ldots, p,
\end{array}%
\end{equation*}%
such that%
\begin{equation*}
\begin{array}{c}
\tilde{g}^\circ_{\beta}\left( \varkappa \right) \cdot g_\circ^{\alpha }\left(
\varkappa \right)=\delta^\alpha_\beta,%
\end{array}%
\end{equation*}%
for any $\varkappa \in V,$ then we say that the $\left( g,h\right) $ is locally
invertible.
\end{definition}
\begin{definition}
If $%
\begin{array}{ccl}
\dot{c}: I \longrightarrow E_{|\func{Im}\left( \eta \circ
h\circ c\right) }%
\end{array}%
$ is a differentiable $\left( g,h\right) $-lift of differentiable curve $c$,
such that its component function $\left(y^\circ\right)
$ is solution for the differentiable system of equations:%
\begin{equation*}
\begin{array}[b]{c}
\frac{du^\circ}{dt}+\left( \rho ,\eta \right) \Gamma^\circ_{\alpha }\circ
u\left( c,\dot{c}\right) \circ \left( \eta \circ h\circ c\right) \cdot
g_\circ^{\alpha }\circ h\circ c\cdot u^\circ=0,%
\end{array}%
\end{equation*}%
then we say that the $\left( g,h\right) $-lift $\dot{c}$ is parallel with respect to the $\left( \rho ,\eta \right) $-connection $\left( \rho ,\eta \right) \Gamma$.
\end{definition}
\subsection{The lift of accelerations for a differentiable curve}
Let
\begin{equation*}
\begin{array}{rcl}
I & ^{\underrightarrow{\ \ \dot{c}\ \ }} & E_{|\func{Im}\left( \eta \circ
h\circ c\right) }\vspace*{1mm},\\
t & \longmapsto & y^{\circ}\left( t\right) s_{\circ}\left( \eta \circ h\circ c\left(
t\right) \right),
\end{array}
\end{equation*}%
be the $\left( g,h\right) $-lift of differentiable curve $%
\begin{array}{ccc}
I & ^{\underrightarrow{\ \ c\ \ }} & M.%
\end{array}%
$
\begin{definition}
The differentiable curve $ \ddot{c}:I \longrightarrow \left( \rho ,\eta \right) TE_{|%
\func{Im}\dot{c}}$ defined by
\begin{equation*}
\begin{array}{rcl}
\ddot{c}(t)=\left( g_\circ^{\alpha }\circ h\circ c\left(
t\right) \cdot y^\circ\left( t\right) \right) \tilde{\partial}_\alpha\left( \dot{c}\left( t\right) \right) +\frac{dy^\circ\left(
t\right) }{dt}\dot{\tilde{\partial}}_\circ\left( \dot{c}\left(
t\right) \right),
\end{array}
\end{equation*}%
is called the differentiable $\left( g,h\right) $-lift of
accelerations of the differentiable curve $c$. Moreover, the section
\begin{equation*}
\begin{array}{rcl}
\func{Im}\left(\dot{c}\right) & ^{\underrightarrow{u\left( c,\dot{c},\ddot{c%
}\right) }} & \left( \rho ,\eta \right) TE_{|\func{Im}\left( \dot{c}\right) },
\vspace*{1mm} \\
\dot{c}\left( t\right) & \longmapsto & \displaystyle\left( g_\circ^{\alpha
}\circ h\circ c\left( t\right) \cdot y^\circ\left( t\right) \right) \tilde{\partial}_\alpha\left( \dot{c}\left( t\right)
\right)+\displaystyle\frac{dy^\circ\left( t\right) }{dt}\dot{\tilde{\partial}}_\circ\left( \dot{c}\left( t\right) \right),
\end{array}%
\end{equation*}%
is called the canonical section associated to the triple $\left(
c,\dot{c},\ddot{c}\right)$.
\end{definition}
In the adapted $\left( \rho ,\eta \right)$-base $\left(\tilde{\delta}_{\alpha }, \dot{\tilde{\partial}}_\circ\right)$, we can rewrite the above equation as follows:
\begin{equation*}
\begin{array}{c}
u\left(c,\dot{c},\ddot{c}\right) \left( \dot{c}\left( t\right) \right) =%
\displaystyle\left( g_{\circ}^{\alpha }\circ h\circ c\left( t\right) y^{\circ}\left(
t\right) \right) \displaystyle\tilde{\delta}_{\alpha }%
\left(\dot{c}\left( t\right) \right) +\displaystyle\frac{dy^{\circ}\left(
t\right) }{dt}\displaystyle\dot{\tilde{\partial}}_\circ\left(
\dot{c}\left( t\right) \right) \vspace*{1mm} \\
\qquad \displaystyle+\left( \rho ,\eta \right) \Gamma _{\alpha }^{\circ}\circ
u\left(c,\dot{c}\right) \circ \eta \circ h\circ c\left( t\right) \cdot
\left(g_{\circ}^{\alpha }\circ h\circ c\left( t\right) y^{\circ}\left( t\right)
\right) \displaystyle\dot{\tilde{\partial}}_\circ\left( \dot{c}%
\left( t\right) \right),
\end{array}
\end{equation*}
where $t\in I$. It is easy to check that $u\left( c,\dot{c},\ddot{c}\right) \left( \dot{c}%
\left( t\right) \right) \in H\left( \rho ,\eta \right) TE_{|\func{Im}\left(
\dot{c}\right) }$ if and only if the component function $y^\circ$ is solution for the differentiable
equations
\begin{equation*}
\begin{array}{c}
\frac{du^{\circ}}{dt}+\left( \rho ,\eta \right) \Gamma _{\alpha }^{\circ}\circ
u\left( c,\dot{c}\right) \circ \eta \circ h\circ c\cdot \left( g_{\circ}^{\alpha
}\circ h\circ c\right) \cdot u^{\circ}=0.
\end{array}
\end{equation*}
\begin{definition}
If the component functions $\left( \left( g_{\circ}^{\alpha }\circ h\circ c\right) y^{\circ}\right)$
are solutions for the differentiable system of equations%
\begin{equation*}
\begin{array}{c}
\frac{dz^{\alpha }}{dt}+\left( \rho ,\eta \right) H_{\beta \gamma }^{\alpha
}\circ u\left( c,\dot{c}\right) \circ \eta \circ h\circ c\cdot z^{\beta
}\cdot z^{\gamma }=0,~\alpha \in 1,\cdots,p,
\end{array}%
\end{equation*}%
then the differentiable curve $\dot{c}$ is called horizontal
parallel with respect to the distinguished linear $\left( \rho ,\eta
\right) $-connection $\left( \left( \rho ,\eta \right) H,\left( \rho
,\eta \right) V\right)$. Also, if the component function $y^\circ$ is
a solution for the differentiable system of equations%
\begin{equation*}
\begin{array}{c}
\frac{du^{\circ}}{dt}+\left( \rho ,\eta \right) V_{\circ\circ}^{\circ}\circ u\left( c,\dot{c}%
\right) \circ \eta \circ h\circ c\cdot u^{\circ}\cdot u^{\circ}=0,
\end{array}
\end{equation*}%
then the differentiable curve $\dot{c}$ is called vertical
parallel with respect to the distinguished linear $\left( \rho ,\eta
\right) $-connection $\left( \left( \rho ,\eta \right) H,\left( \rho
,\eta \right) V\right)$.
\end{definition}
\section{The $\left( \protect\rho ,\protect\eta ,h\right) $-torsion and the $%
\left( \protect\rho ,\protect\eta ,h\right) $-curvature of a distinguished
linear $\left( \protect\rho ,\protect\eta \right) $-connection}
Let $(\rho, \eta)\Gamma$ be a $(\rho, \eta)$-connection for the Kaluza-Klein bundle $(E, \pi, M)$ and let $((\rho, \eta)H, (\rho, \eta)V)$ be a distinguished linear $(\rho, \eta)$-connection for the Lie algebroid generalized tangent bundle
\[
\begin{array}{c}
\left( \left( \left( \rho ,\eta \right) TE,\left( \rho ,\eta \right) \tau
_{E},E\right) ,\left[ ,\right] _{\left( \rho ,\eta \right) TE},\left( \tilde{%
\rho},Id_{E}\right) \right).
\end{array}
\]
\begin{definition}
The application
\begin{equation*}
\begin{array}{rcl}
\Gamma \left( \left( \rho ,\eta \right) TE,\left( \rho ,\eta \right) \tau
_{E},E\right) ^{2} & ^{\underrightarrow{\left( \rho ,\eta ,h\right) \mathbb{T%
}}} & \Gamma \left( \left( \rho ,\eta \right) TE,\left( \rho ,\eta \right)
\tau _{E},E\right) \vspace*{2mm},\\
\left( X,Y\right) & \longmapsto & \left( \rho ,\eta \right) \mathbb{T}\left(
X,Y\right),
\end{array}%
\end{equation*}%
defined by
\begin{equation*}
\begin{array}{c}
\left( \rho ,\eta ,h\right) \mathbb{T}\left( X,Y\right) =\left( \rho ,\eta
\right) D_{X}Y-\left( \rho ,\eta \right) D_{Y}X-\left[ X,Y\right] _{\left(
\rho ,\eta \right) TE},%
\end{array}
\end{equation*}%
for any $X,Y\in \Gamma \left( \left( \rho ,\eta \right) TE,\left( \rho ,\eta
\right) \tau _{E},E\right)$, is called the $\left( \rho ,\eta
,h\right) $-torsion associated to distinguished linear $\left( \rho
,\eta \right) $-connection $\left( \left( \rho ,\eta \right)
H,\left( \rho ,\eta \right) V\right) .$
Moreover, the applications
\begin{equation*}
\mathcal{H}\left( \rho ,\eta ,h\right) \mathbb{T}\left( \mathcal{H}\left(
\cdot \right) ,\mathcal{H}\left( \cdot \right) \right) ,\,\mathcal{V}\left(
\rho ,\eta ,h\right) \mathbb{T}\left( \mathcal{H}\left( \cdot \right) ,%
\mathcal{H}\left( \cdot \right) \right), \ldots, \mathcal{V}\left( \rho ,\eta
,h\right) \mathbb{T}\left( \mathcal{V}\left( \cdot \right) ,\mathcal{V}
\left( \cdot \right) \right)
\end{equation*}%
are called $\mathcal{H}\left( \mathcal{HH}\right) ,\,\mathcal{V}\left(
\mathcal{HH}\right) , \ldots, \mathcal{V}\left( \mathcal{VV}\right) $ $\left(
\rho ,\eta ,h\right) $-torsions associated to distinguished linear $%
\left( \rho ,\eta \right) $-connection $\left( \left( \rho ,\eta
\right) H,\left( \rho ,\eta \right) V\right) .$
\end{definition}
\begin{proposition}
The $\left( \rho ,\eta
,h\right) $-torsion $\left( \rho ,\eta ,h\right) \mathbb{T}$ associated to distinguished linear $\left( \rho ,\eta \right) $-connection $\left( \left( \rho ,\eta \right) H,\left( \rho ,\eta \right)
V\right) $, is $\mathbb{R}$-bilinear and antisymmetric in the
lower indices.
\end{proposition}
Using the notations
\begin{equation*}
\begin{array}{l}
\mathcal{H}\left( \rho ,\eta ,h\right) \mathbb{T}\left( \tilde{\delta}%
_{\gamma },\tilde{\delta}_{\beta }\right) =\left( \rho ,\eta ,h\right)
\mathbb{T}_{~\beta \gamma }^{\alpha }\tilde{\delta}_{\alpha },\vspace*{1mm}
\\
\mathcal{V}\left( \rho ,\eta ,h\right) \mathbb{T}\left( \tilde{\delta}%
_{\gamma },\tilde{\delta}_{\beta }\right) =\left( \rho ,\eta ,h\right)
\mathbb{T}_{~\beta \gamma }^{\circ}\dot{\tilde{\partial}}_{\circ},%
\vspace*{1mm} \\
\mathcal{H}\left( \rho ,\eta ,h\right) \mathbb{T}\left( \dot{%
\tilde{\partial}}_{\circ},\tilde{\delta}_{\beta }\right) =\left( \rho ,\eta
,h\right) \mathbb{P}_{~\beta \circ}^{\alpha }\tilde{\delta}_{\alpha },\vspace*{%
1mm} \\
\mathcal{V}\left( \rho ,\eta ,h\right) \mathbb{T}\left( \dot{%
\tilde{\partial}}_{\circ},\tilde{\delta}_{\beta }\right) =\left( \rho ,\eta
,h\right) \mathbb{P}_{~\beta \circ}^{\circ}\dot{\tilde{\partial}}_{\circ},%
\vspace*{1mm} \\
\ \mathcal{V}\left( \rho ,\eta ,h\right) \mathbb{T}\left( \dot{%
\tilde{\partial}}_{\circ},\dot{\tilde{\partial}}_{\circ}\right) =\left(
\rho ,\eta ,h\right) \mathbb{S}_{~\circ\circ}^{\circ}\dot{\tilde{\partial}}%
_{\circ},
\end{array}%
\end{equation*}
we can prove the following
\begin{theorem}
The $\left( \rho ,\eta
,h\right) $-torsion $\left( \rho ,\eta ,h\right) \mathbb{T}$ associated to the distinguished linear $\left( \rho ,\eta \right) $-connection $\left( \left( \rho ,\eta \right) H,\left( \rho ,\eta \right)
V\right) $, is characterized by the tensor fields with local
components:
\begin{equation*}
\begin{array}{cl}
\left( \rho ,\eta ,h\right) \mathbb{T}_{~\beta \gamma }^{\alpha } & =\left(
\rho ,\eta \right) H_{\beta \gamma }^{\alpha }-\left( \rho ,\eta \right)
H_{\gamma \beta }^{\alpha }-L_{\beta \gamma }^{\alpha }\circ h\circ \pi ,%
\vspace*{2mm} \\
\left( \rho ,\eta ,h\right) \mathbb{T}_{\ \beta \gamma }^{\circ} & =\left( \rho
,\eta ,h\right) \mathbb{R}_{\,\ \beta \gamma }^{\circ},\vspace*{2mm} \\
\left( \rho ,\eta ,h\right) \mathbb{P}_{~\beta \circ}^{\alpha } & =\left( \rho
,\eta \right) V_{\beta \circ}^{\alpha },\vspace*{2mm} \\
\left( \rho ,\eta ,h\right) \mathbb{P}_{~\beta \circ}^{\circ} & =\displaystyle\frac{%
\partial }{\partial y^{\circ}}\left( \left( \rho ,\eta \right) \Gamma _{\beta
}^{\circ}\right) -\left( \rho ,\eta \right) H_{\circ\beta }^{\circ},\vspace*{2mm} \\
\left( \rho ,\eta ,h\right) \mathbb{S}_{~\circ\circ}^{\circ} & =\left( \rho ,\eta
\right) V_{\circ\circ}^{\circ}-\left( \rho ,\eta \right) V_{\circ\circ}^{\circ}=0.%
\end{array}%
\end{equation*}
\end{theorem}
In particular, when $\left( \rho ,\eta ,h\right) =\left(
Id_{TM},Id_{M},Id_{M}\right)$, we deduce the following local components of
torsion associated to distinguished linear connection $\left( H,V\right)$:
\begin{equation*}
\begin{array}{clcl}
\mathbb{T}_{~jk}^{i} & =H_{jk}^{i}-H_{kj}^{i}, & \mathbb{T}_{~jk}^{\circ} & =%
\mathbb{R}_{\,\ jk}^{\circ},\vspace*{1mm} \\
\mathbb{P}_{~j\circ}^{i} & =V_{j\circ}^{i}, & \mathbb{P}_{~j\circ}^{\circ} & =\displaystyle%
\frac{\partial \Gamma _{j}^{\circ}}{\partial y^{\circ}}-H_{\circ j}^{\circ},\vspace*{1mm} \\
\mathbb{S}_{~\circ\circ}^{\circ} & =V_{\circ\circ}^{\circ}-V_{\circ\circ}^{\circ}=0. &  &
\end{array}%
\end{equation*}
\begin{definition}
The application
\begin{equation*}
\begin{array}{rcc}
\left( \Gamma \left( \left( \rho ,\eta \right) TE,\left( \rho ,\eta \right)
\tau _{E},E\right) \right) ^{3} & ^{\underrightarrow{\ \left( \rho ,\eta
,h\right) \mathbb{R}\ }} & \Gamma \left( \left( \rho ,\eta \right) TE,\left(
\rho ,\eta \right) \tau _{E},E\right) \vspace*{3mm},\\
\left( \left( Y,Z\right) ,X\right) & \longmapsto & \left( \rho ,\eta
,h\right) \mathbb{R}\left( \left( Y,Z\right) ,X\right),
\end{array}
\end{equation*}%
defined by
\begin{align*}
\left( \rho ,\eta ,h\right) \mathbb{R}\left( Y,Z\right) X&=\left( \rho ,\eta
\right) D_{Y}\left( \left( \rho ,\eta \right) D_{Z}X\right)-\left( \rho ,\eta \right) D_{Z}\left( \left( \rho ,\eta \right)
D_{Y}X\right)\\
&\ \ \ -\left( \rho ,\eta \right) D_{\left[ Y,Z\right] _{\left( \rho
,\eta \right) TE}}X,
\end{align*}
for any $X,Y,Z\in \Gamma \left( \left( \rho ,\eta \right) TE,\left( \rho
,\eta \right) \tau _{E},E\right)$, is called the $\left( \rho ,\eta
,h\right) $-curvature associated to distinguished linear $\left(
\rho ,\eta \right) $-connection $\left( \left( \rho ,\eta \right)
H,\left( \rho ,\eta \right) V\right) .$
\end{definition}
\begin{proposition}
The $\left( \rho ,\eta
,h\right) $-curvature $\left( \rho ,\eta ,h\right) \mathbb{R}$ associated to distinguished linear $\left( \rho ,\eta \right) $-connection $\left( \left( \rho ,\eta \right) H,\left( \rho ,\eta \right)
V\right) $, is $\mathbb{R}$-linear in each argument and
antisymmetric in the first two arguments.
\end{proposition}
Using the notations
\begin{equation*}
\begin{array}{rl}
\left( \rho ,\eta ,h\right) \mathbb{R}\left( \tilde{\delta}_{\varepsilon },%
\tilde{\delta}_{\gamma }\right) \tilde{\delta}_{\beta } & =\left( \rho ,\eta
,h\right) \mathbb{R}_{\ \beta\ \gamma \varepsilon }^{\alpha }\tilde{\delta}%
_{\alpha }, \\
\left( \rho ,\eta ,h\right) \mathbb{R}\left( \tilde{\delta}_{\varepsilon },%
\tilde{\delta}_{\gamma }\right) \dot{\tilde{\partial}}_{\circ} &
=\left( \rho ,\eta ,h\right) \mathbb{R}_{\ \circ\ \gamma \varepsilon }^{\circ}\dot{\tilde{\partial}}_{\circ}, \\
\left( \rho ,\eta ,h\right) \mathbb{R}\left( \dot{\tilde{\partial%
}}_{\circ},\tilde{\delta}_{\gamma }\right) \tilde{\delta}_{\varepsilon } &
=\left( \rho ,\eta ,h\right) \mathbb{P}_{\ \varepsilon\ \gamma \circ}^{\alpha }%
\tilde{\delta}_{\alpha }, \\
\left( \rho ,\eta ,h\right) \mathbb{R}\left( \dot{\tilde{\partial%
}}_{\circ},\tilde{\delta}_{\gamma }\right) \dot{\tilde{\partial}}_{\circ}
& =\left( \rho ,\eta ,h\right) \mathbb{P}_{\ \circ\ \gamma \circ}^{\circ}\dot{%
\tilde{\partial}}_{\circ}, \\
\left( \rho ,\eta ,h\right) \mathbb{R}\left( \dot{\tilde{\partial%
}}_{\circ},\dot{\tilde{\partial}}_{\circ}\right) \tilde{\delta}_{\beta }
& =\left( \rho ,\eta ,h\right) \mathbb{S}_{\ \beta\ \circ\circ}^{\alpha }\tilde{\delta%
}_{\alpha },\vspace*{1mm} \\
\left( \rho ,\eta ,h\right) \mathbb{R}\left( \dot{\tilde{\partial%
}}_{\circ},\dot{\tilde{\partial}}_{\circ}\right) \dot{\tilde{%
\partial}}_{\circ} & =\left( \rho ,\eta ,h\right) \mathbb{S}_{\ \circ\ \circ\circ}^{\circ}\dot{\tilde{\partial}}_{\circ},
\end{array}
\end{equation*}
we derive the following
\begin{theorem}
The $\left(\rho ,\eta
, h\right) $-curvature $\left( \rho ,\eta ,h\right) \mathbb{R}$ associated to distinguished linear $\left( \rho ,\eta \right) $-connection $\left( \left( \rho ,\eta \right) H,\left( \rho ,\eta \right)
V\right)$, is characterized by the\ tensor fields with the local
components:
\begin{equation*}
\left\{
\begin{array}{cl}
(\rho ,\eta ,h)\mathbb{R}_{\ \beta\ \gamma \varepsilon }^{\alpha }\!\!\!\!\!&=%
\Gamma (\tilde{\rho},Id_{E})\!\left( \!\tilde{\delta}_{\varepsilon
}\!\right) \!(\rho ,\eta )H_{\beta \gamma }^{\alpha }{-}\Gamma \!(\tilde{\rho%
},Id_{E})\!\left( \!\tilde{\delta}_{\gamma }\!\right) (\rho ,\eta )H_{\beta
\varepsilon }^{\alpha }\vspace*{2mm} \\
& +(\rho ,\eta )H_{\theta \varepsilon }^{\alpha }(\rho ,\eta )H_{\beta
\gamma }^{\theta }{-}(\rho ,\eta )H_{\theta \gamma }^{\alpha }(\rho ,\eta
)H_{\beta \varepsilon }^{\theta }\vspace*{2mm} \\
& -(\rho ,\eta ,h)\mathbb{R}_{\ \gamma \varepsilon }^{\circ}(\rho ,\eta
)H_{\beta \circ}^{\alpha }{-}L_{\gamma \varepsilon }^{\theta }\circ h\circ \pi
(\rho ,\eta )H_{\beta \theta }^{\alpha },\vspace*{3mm} \\
(\rho ,\eta ,h)\mathbb{R}_{\ \circ\ \gamma\varepsilon }^{\circ}\!\!\!\!\!&=\Gamma (\tilde{%
\rho}, Id_{E})\!\left( \!\tilde{\delta}_{\varepsilon }\!\right) \!(\rho ,\eta
)H_{\circ\gamma }^{\circ}{-}\Gamma (\tilde{\rho}, Id_{E})\!\left( \!\tilde{\delta}%
_{\gamma }\!\right) (\rho ,\eta )H_{\circ\varepsilon }^{\circ}\vspace*{2mm} \\
& +(\rho ,\eta )H_{\circ\varepsilon }^{\circ}(\rho ,\eta )H_{\circ\gamma }^{\circ}{-}(\rho
,\eta )H_{\circ\gamma }^{\circ}(\rho ,\eta )H_{\circ\varepsilon }^{\circ}\vspace*{2mm} \\
& -(\rho ,\eta ,h)\mathbb{R}_{\ \varepsilon \gamma }^{\circ}(\rho ,\eta
)V_{\circ\circ}^{\circ}{-}L_{\gamma \varepsilon }^{\theta }\circ h\circ \pi (\rho ,\eta
)V_{\circ\theta }^{\circ},%
\end{array}%
\right.
\end{equation*}%
\begin{equation*}
\newline
\left\{
\begin{array}{ll}
(\rho ,\eta ,h)\mathbb{P}_{\ \varepsilon\ \gamma\circ}^{\alpha}\!\!\!& =\Gamma
(\tilde{\rho}, Id_{E})(\dot{\tilde{\partial}}_{\circ})(\rho ,\eta
)H_{\varepsilon \gamma }^{\alpha }{-}\vspace*{1mm}\Gamma (\tilde{\rho}%
,Id_{E})\!\left( \!\tilde{\delta}_{\gamma }\!\right) \!(\rho ,\eta
)V_{\varepsilon \circ}^{\alpha }\vspace*{2mm} \\
& {+}(\rho ,\eta )V_{\theta \circ}^{\alpha }(\rho ,\eta )H_{\varepsilon \gamma
}^{\theta }-\vspace*{1mm}(\rho ,\eta )H_{\theta \gamma }^{\alpha }(\rho
,\eta )V_{\varepsilon \circ}^{\theta }\vspace*{2mm} \\
& +\displaystyle\dot{\partial}_\circ\left( \left( \rho ,\eta
\right) \Gamma _{\gamma }^{\circ}\right) \left( \rho ,\eta \right)
V_{\varepsilon \circ}^{\alpha },\vspace*{3mm} \\
(\rho ,\eta ,h)\mathbb{P}_{\ \circ\ \gamma\circ}^{\circ}\!\!\! & =\Gamma (\tilde{\rho}%
,Id_{E})\!\left( \!\dot{\tilde{\partial}}_{\circ}\!\right) \!(\rho
,\eta )H_{\circ\gamma }^{\circ}-\vspace*{2mm} \\
& -\Gamma (\tilde{\rho},Id_{E})\!\left( \!\tilde{\delta}_{\gamma }\!\right)
\!(\rho ,\eta )V_{\circ\circ}^{\circ}+(\rho ,\eta )V_{\circ\circ}^{\circ}(\rho ,\eta )H_{\circ\gamma
}^{\circ}-\vspace*{2mm} \\
& -(\rho ,\eta )H_{\circ\gamma }^{\circ}(\rho ,\eta )V_{\circ\circ}^{\circ}+\displaystyle\dot{\partial}_\circ(\left( \rho ,\eta )\Gamma _{\gamma }^{\circ}\right)
(\rho ,\eta )V_{\circ\circ}^{\circ},%
\end{array}%
\right. \newline
\end{equation*}%
$\newline
\newline
$%
\begin{equation*}
\left\{
\begin{array}{cl}
\left(\rho, \eta, h\right) \mathbb{S}_{\ \beta\ \circ\circ}^{\alpha }\!\!\!\!\!\!&=\Gamma
\left( \tilde{\rho},Id_{E}\right) \left(\dot{\tilde{\partial}}%
_{\circ}\right) \left(\rho, \eta \right) V_{\beta \circ}^{\alpha }\vspace*{2mm} \\
& -\Gamma \left( \tilde{\rho},Id_{E}\right) \left( \dot{\tilde{%
\partial}}_{\circ}\right) \left( \rho ,\eta \right) V_{\beta \circ}^{\alpha }+\left(
\rho ,\eta \right) V_{\theta \circ}^{\alpha }\left( \rho ,\eta \right) V_{\beta
\circ}^{\theta }\vspace*{2mm} \\
& -\left( \rho ,\eta \right) V_{\theta \circ}^{\alpha }\left( \rho ,\eta \right)
V_{\beta \circ}^{\theta }=0,\vspace*{3mm} \\
\left( \rho ,\eta ,h\right) \mathbb{S}_{\ \circ\ \circ\circ}^{\circ}\!\!\!\!\!\!&=\Gamma \left( \tilde{%
\rho},Id_{E}\right) \left( \dot{\tilde{\partial}}_{\circ}\right)
\left( \rho ,\eta \right) V_{\circ\circ}^{\circ}\vspace*{2mm} \\
& -\Gamma \left( \tilde{\rho},Id_{E}\right) \left( \dot{\tilde{%
\partial}}_{\circ}\right) \left( \rho ,\eta \right) V_{\circ\circ}^{\circ}+\left( \rho ,\eta
\right) V_{\circ\circ}^{\circ}\left( \rho ,\eta \right) V_{\circ\circ}^{\circ}\vspace*{2mm} \\
& -\left( \rho ,\eta \right) V_{\circ\circ}^{\circ}\left( \rho ,\eta \right) V_{\circ\circ}^{\circ}=0.%
\end{array}%
\right. \bigskip
\end{equation*}
\end{theorem}
In particular, when $\left( \rho ,\eta ,h\right) =\left(
Id_{TM},Id_{M},Id_{M}\right)$, we obtain the following local components of the
curvature associated to distinguished linear connection $\left( H,V\right)$:
\begin{equation*}
\begin{array}{l}
\mathbb{R}_{\ j\ kl}^{i}=\delta _{l}\left( H_{jk}^{i}\right) -\delta
_{k}\left( H_{jl}^{i}\right) +H_{hl}^{i}H_{jk}^{h}-H_{hk}^{i}H_{jl}^{h}-%
\mathbb{R}_{\,\ kl}^{\circ}H_{j\circ}^{i},\vspace*{1mm} \\
\mathbb{R}_{\ \circ\ kl}^{\circ}=\delta _{l}\left( H_{\circ k}^{\circ}\right) -\delta
_{k}\left( H_{\circ l}^{\circ}\right) +H_{\circ l}^{\circ}H_{\circ k}^{\circ}-H_{\circ k}^{\circ}H_{\circ l}^{\circ}-%
\mathbb{R}_{\,\ lk}^{\circ}V_{\circ\circ}^{\circ},%
\end{array}
\end{equation*}
\begin{equation*}
\begin{array}{l}
\mathbb{P}_{\ l\ k\circ}^{i}=\displaystyle\dot{\partial}_\circ\left(
H_{lk}^{i}\right)-\delta _{k}\left( V_{l\circ}^{i}\right)
+V_{h\circ}^{i}H_{lk}^{h}-H_{hk}^{i}V_{l\circ}^{h}+\dot{\partial}_\circ%
\left( \Gamma _{k}^{\circ}\right) V_{l\circ}^{i},\vspace*{2mm} \\
\mathbb{P}_{\ \circ\ k\circ}^{\circ}=\displaystyle\dot{\partial}_\circ\left(
H_{\circ k}^{\circ}\right)-\delta _{k}\left( V_{\circ\circ}^{\circ}\right)
+V_{\circ\circ}^{\circ}H_{\circ k}^{\circ}-H_{\circ k}^{\circ}V_{\circ\circ}^{\circ}+\dot{\partial}_\circ%
\left(\Gamma _{k}^{\circ}\right) V_{\circ\circ}^{\circ},%
\end{array}
\end{equation*}
\begin{equation*}
\begin{array}{l}
\mathbb{S}_{\ j\ \circ\circ}^{i}=\displaystyle\dot{\partial}_\circ\left(
V_{j\circ}^{i}\right)-\dot{\partial}_\circ\left( V_{j\circ}^{i}\right)
+V_{h\circ}^{i}V_{j\circ}^{h}-V_{h\circ}^{i}V_{j\circ}^{h}=0,\vspace*{2mm} \\
\mathbb{S}_{\ \circ\ \circ\circ}^{\circ}=\displaystyle\dot{\partial}_\circ\left(
V_{\circ\circ}^{\circ}\right) -\dot{\partial}_\circ\left( V_{\circ\circ}^{\circ}\right)
+V_{\circ\circ}^{\circ}V_{\circ\circ}^{\circ}-V_{\circ\circ}^{\circ}V_{\circ\circ}^{\circ}=0.%
\end{array}
\end{equation*}
\begin{definition}
The tensor field
\begin{align*}
\mathbf{Ric}\left( \left( \rho ,\eta \right) H,\left( \rho ,\eta \right)
V\right)&=\left( \rho ,\eta ,h\right) \mathbb{R}_{\ \alpha \ \beta }d\tilde{z}%
^{\alpha }\otimes d\tilde{z}^{\beta }+\left( \rho ,\eta ,h\right) \mathbb{P}%
_{\ \alpha\ \circ}d\tilde{z}^{\alpha }\otimes \delta \tilde{y}^{\circ}\vspace*{2mm}
\\
&\ \ +\left( \rho ,\eta ,h\right) \mathbb{P}_{\ \circ\ \beta }\delta \tilde{y}%
^{\circ}\otimes d\tilde{z}^{\beta }+\left( \rho ,\eta ,h\right) \mathbb{S}%
_{\ \circ\ \circ}\delta \tilde{y}^{\circ}\otimes \delta \tilde{y}^{\circ},
\end{align*}
is called the Ricci tensor field associated to distinguished
linear $\left( \rho ,\eta \right) $-connection $\left( \left(\rho
,\eta \right) H,\left( \rho ,\eta \right) V\right)$, where
\begin{equation*}
\begin{array}{ll}
\left(\rho ,\eta ,h\right) \mathbb{R}_{\ \alpha \ \beta }=\left( \rho ,\eta
,h\right) \mathbb{R}_{\ \alpha \ \beta \gamma }^{\gamma}, & \left( \rho ,\eta
,h\right) \mathbb{P}_{\ \alpha \ \circ}=\left( \rho ,\eta ,h\right) \mathbb{P}%
_{\ \alpha \ \beta \circ}^{\beta }\vspace*{2mm},\\
\left( \rho ,\eta ,h\right) \mathbb{P}_{\ \circ\ \beta}=\left( \rho ,\eta
,h\right) \mathbb{P}_{\ \circ\ \beta \circ}^{\circ}, & \left( \rho ,\eta ,h\right) \mathbb{S%
}_{\ \circ\ \circ}=\left( \rho ,\eta ,h\right) \mathbb{S}_{\ \circ\ \circ\circ}^{\circ}.
\end{array}%
\end{equation*}
\end{definition}
\subsection{Identities of Cartan and Bianchi type}
Let $(\rho, \eta)\Gamma$ be a $(\rho, \eta)$-connection for the Kaluza-Klein bundle $(E, \pi, M)$ and let $((\rho, \eta)H, (\rho, \eta)V)$ be a distinguished linear $(\rho, \eta)$-connection for the Lie algebroid generalized tangent bundle
\[
\begin{array}{c}
\left( \left( \left( \rho ,\eta \right) TE,\left( \rho ,\eta \right) \tau
_{E},E\right) ,\left[ ,\right] _{\left( \rho ,\eta \right) TE},\left( \tilde{%
\rho},Id_{E}\right) \right).
\end{array}
\]
Using the definition of
$\left( \rho ,\eta ,h\right) $-curvature associated to the
distinguished linear $\left( \rho ,\eta \right) $-connection $\left(
\left( \rho ,\eta \right) H,\left( \rho ,\eta \right) V\right)$, we have the following
\begin{theorem}
The following equalities hold
\begin{equation}\label{R1}
\left\{
\begin{array}{l}
\begin{array}{l}
\left( \rho ,\eta \right) D_{\mathcal{H}X}\left( \rho ,\eta \right) D_{%
\mathcal{H}Y}\mathcal{H}Z-\left( \rho ,\eta \right) D_{\mathcal{H}Y}\left(
\rho ,\eta \right) D_{\mathcal{H}X}\mathcal{H}Z \\
=\left( \rho ,\eta ,h\right) \mathbb{R}\left( \mathcal{H}X,\mathcal{H}%
Y\right) \mathcal{H}Z+\left( \rho ,\eta \right) D_{\mathcal{H}\left[
\mathcal{H}X,\mathcal{H}Y\right] _{\left( \rho ,\eta \right) TE}}\mathcal{H}Z
\\
+\left( \rho ,\eta \right) D_{\mathcal{V}\left[ \mathcal{H}X,\mathcal{H}Y%
\right] _{\left( \rho ,\eta \right) TE}}\mathcal{H}Z,%
\end{array}
\\
\multicolumn{1}{c}{%
\begin{array}{l}
\left( \rho ,\eta \right) D_{\mathcal{V}X}\left( \rho ,\eta \right) D_{%
\mathcal{H}Y}\mathcal{H}Z-\left( \rho ,\eta \right) D_{\mathcal{H}Y}\left(
\rho ,\eta \right) D_{\mathcal{V}X}\mathcal{H}Z \\
=\left( \rho ,\eta ,h\right) \mathbb{R}\left( \mathcal{V}X,\mathcal{H}%
Y\right) \mathcal{H}Z+\left( \rho ,\eta \right) D_{\mathcal{H}\left[
\mathcal{V}X,\mathcal{H}Y\right] _{\left( \rho ,\eta \right) TE}}\mathcal{H}Z
\\
+\left( \rho ,\eta \right) D_{\mathcal{V}\left[ \mathcal{V}X,\mathcal{H}Y%
\right] _{\left( \rho ,\eta \right) TE}}\mathcal{H}Z,%
\end{array}%
} \\
\begin{array}{c}
\left( \rho ,\eta \right) D_{\mathcal{V}X}\left( \rho ,\eta \right) D_{%
\mathcal{V}Y}\mathcal{H}Z-\left( \rho ,\eta \right) D_{\mathcal{V}Y}\left(
\rho ,\eta \right) D_{\mathcal{V}X}\mathcal{H}Z \\
=\left( \rho ,\eta ,h\right) \mathbb{R}\left( \mathcal{V}X,\mathcal{V}%
Y\right) \mathcal{H}Z+\left( \rho ,\eta \right) D_{\mathcal{V}\left[
\mathcal{V}X,\mathcal{V}Y\right] _{\left( \rho ,\eta \right) TE}}\mathcal{H}%
Z,%
\end{array}%
\end{array}%
\right.
\end{equation}%
and
\begin{equation}\label{R2}
\left\{
\begin{array}{c}
\begin{array}{l}
\left( \rho ,\eta \right) D_{\mathcal{H}X}\left( \rho ,\eta \right) D_{%
\mathcal{H}Y}\mathcal{V}Z-\left( \rho ,\eta \right) D_{\mathcal{H}Y}\left(
\rho ,\eta \right) D_{\mathcal{H}X}\mathcal{V}Z \\
=\left( \rho ,\eta ,h\right) \mathbb{R}\left( \mathcal{H}X,\mathcal{H}%
Y\right) \mathcal{V}Z+\left( \rho ,\eta \right) D_{\mathcal{H}\left[
\mathcal{H}X,\mathcal{H}Y\right] _{\left( \rho ,\eta \right) TE}}\mathcal{V}Z
\\
+\left( \rho ,\eta \right) D_{\mathcal{V}\left[ \mathcal{H}X,\mathcal{H}Y%
\right] _{\left( \rho ,\eta \right) TE}}\mathcal{V}Z,%
\end{array}
\\
\begin{array}{l}
\left( \rho ,\eta \right) D_{\mathcal{V}X}\left( \rho ,\eta \right) D_{%
\mathcal{H}Y}\mathcal{V}Z-\left( \rho ,\eta \right) D_{\mathcal{H}Y}\left(
\rho ,\eta \right) D_{\mathcal{V}X}\mathcal{V}Z \\
=\left( \rho ,\eta ,h\right) \mathbb{R}\left( \mathcal{V}X,\mathcal{H}%
Y\right) \mathcal{V}Z+\left( \rho ,\eta \right) D_{h\left[ \mathcal{V}X,%
\mathcal{H}Y\right] _{\left( \rho ,\eta \right) TE}}\mathcal{V}Z \\
+\left( \rho ,\eta \right) D_{\mathcal{V}\left[ \mathcal{V}X,\mathcal{H}Y%
\right] _{\left( \rho ,\eta \right) TE}}\mathcal{V}Z,%
\end{array}
\\
\begin{array}{l}
\left( \rho ,\eta \right) D_{\mathcal{V}X}\left( \rho ,\eta \right) D_{%
\mathcal{V}Y}\mathcal{V}Z-\left( \rho ,\eta \right) D_{\mathcal{V}Y}\left(
\rho ,\eta \right) D_{\mathcal{V}X}\mathcal{V}Z \\
=\left( \rho ,\eta ,h\right) \mathbb{R}\left( \mathcal{V}X,\mathcal{V}%
Y\right) \mathcal{V}Z+\left( \rho ,\eta \right) D_{\mathcal{V}\left[
\mathcal{V}X,\mathcal{V}Y\right] _{\left( \rho ,\eta \right) TE}}\mathcal{V}%
Z.%
\end{array}%
\end{array}%
\right.
\end{equation}
\end{theorem}
Using the above theorem, we obtain the following
formulas of Ricci type:
\begin{equation}\label{R1'}
\left\{
\begin{array}{l}
\begin{array}{cl}
\tilde{Z}_{\ \ |\gamma |\beta }^{\alpha }-\tilde{Z}_{\ \ |\beta |\gamma
}^{\alpha } & =\left( \rho ,\eta ,h\right) \mathbb{R}_{\ \theta\ \gamma \beta
}^{\mathbb{\alpha }}\tilde{Z}^{\theta }+\left( L_{\beta \gamma }^{\theta
}\circ h\circ \pi \right) \tilde{Z}_{\ \ |\theta }^{\alpha }\vspace*{2mm} \\
& +\left( \rho ,\eta ,h\right) \mathbb{T}_{\ \beta \gamma }^{\circ}\tilde{Z}%
^{\alpha }|_{\circ}+\left( \rho ,\eta ,h\right) \mathbb{T}_{\ \beta \gamma
}^{\theta }\tilde{Z}_{\ \ |\theta }^{\alpha},%
\end{array}%
\vspace*{2mm} \\
\begin{array}{cl}
\tilde{Z}_{\ \ |\gamma }^{\alpha }|_{\circ}-\tilde{Z}^{\alpha }|_{\circ}{}_{|\gamma }
& =\left( \rho ,\eta ,h\right) \mathbb{P}_{\ \theta \ \gamma \circ}^{\mathbb{%
\alpha }}\tilde{Z}^{\theta }-\left( \rho ,\eta ,h\right) \mathbb{P}_{\ \gamma
\circ}^{\circ}\tilde{Z}^{\alpha }\vspace*{2mm}|_{\circ} \\
& -\left( \rho ,\eta \right) \mathbb{H}_{\circ\gamma }^{\circ}\tilde{Z}^{\alpha
}|_{\circ},%
\end{array}%
\vspace*{2mm} \\
\begin{array}{cc}
\tilde{Z}^{\alpha }|_{\circ}|_{\circ}-\tilde{Z}^{\alpha }|_{\circ}|_{\circ} & =\left( \rho
,\eta ,h\right) \mathbb{S}_{\ \theta \ \circ\circ}^{\mathbb{\alpha }}\tilde{Z}^{\theta
}+\left( \rho ,\eta ,h\right) \mathbb{S}_{\ \circ\circ}^{\circ}\tilde{Z}^{\alpha }|_{\circ},%
\end{array}%
\end{array}%
\right.
\end{equation}
and
\begin{equation}\label{R2'}
\left\{
\begin{array}{l}
\begin{array}{cl}
Y_{\ \ |\gamma |\beta }^{\circ}-Y_{\ \ |\beta |\gamma }^{\circ} & =\left( \rho ,\eta
,h\right) \mathbb{R}_{\ \circ\ \gamma \beta }^{\circ}Y^{\circ}+\left( L_{\beta \gamma
}^{\theta }\circ h\circ \pi \right) Y_{\ \ |\theta }^{\circ}\vspace*{2mm} \\
& +\left( \rho ,\eta \right) \mathbb{T}_{\ \beta \gamma
}^{\circ}Y^{\circ}|_{\circ}+\left( \rho ,\eta ,h\right) \mathbb{T}_{\ \beta \gamma
}^{\theta }Y_{\ \ |\theta }^{\circ},%
\end{array}%
\vspace*{2mm} \\
\begin{array}{cl}
Y_{\ \ |\gamma }^{\circ}|_{\circ}-Y^{\circ}|_{\circ}{}_{|\gamma } & =\left( \rho ,\eta
,h\right) \mathbb{P}_{\ \theta \ \gamma \circ}^{\circ}Y^{\theta}-\left( \rho ,\eta
,h\right) \mathbb{P}_{\ \gamma \circ}^{\circ}Y^{\circ}\vspace*{2mm}|_{\circ} \\
& -\left( \rho ,\eta \right) \mathbb{H}_{\circ\gamma }^{\circ}Y^{\circ}|_{\circ},%
\end{array}%
\vspace*{2mm} \\
\begin{array}{cc}
Y^{\circ}|_{\circ}|_{\circ}-Y^{\circ}|_{\circ}|_{\circ} & =\left( \rho ,\eta ,h\right) \mathbb{S}%
_{\ \circ\ \circ\circ}^{\circ}Y^{\circ}+\left( \rho ,\eta ,h\right) \mathbb{S}_{\ \circ\circ}^{\circ}Y^{\circ}|_{\circ}.%
\end{array}%
\end{array}%
\right.
\end{equation}
In particular, if $\left( \rho ,\eta ,h\right) =\left(
Id_{TM},Id_{M},id_{M}\right) $ and the Lie bracket $\left[ ,\right]
_{TM}$ is the usual Lie bracket, then the formulas of Ricci type (\ref{R1'}) and (\ref{R2'}) reduce to
\begin{equation*}
\left\{
\begin{array}{cl}
\tilde{Z}_{\ \ |k|j}^{i}-\tilde{Z}_{\ \ |j|k}^{i} & =\mathbb{R}_{\ h\ kj}^{i}%
\tilde{Z}^{h}+\mathbb{T}_{\ jk}^{\circ}\tilde{Z}^{i}|_{\circ}+\mathbb{T}_{\ jk}^{h}%
\tilde{Z}_{\ \ |h}^{i},\vspace*{2mm} \\
\tilde{Z}^{i}\ _{|k}|_{\circ}-\tilde{Z}^{i}|_{\circ|k} & =\mathbb{P}_{\ h\ k\circ}^{i}%
\tilde{Z}^{h}-\mathbb{P}_{\ k\circ}^{\circ}\tilde{Z}^{i}|_{\circ}-\mathbb{H}_{\circ k}^{\circ}%
\tilde{Z}^{i}|_{\circ},\vspace*{2mm} \\
\tilde{Z}^{i}|_{\circ}|_{\circ}-\tilde{Z}^{i}|_{\circ}|_{\circ} & =\mathbb{S}_{\ h\ \circ\circ}^{i}%
\tilde{Z}^{h}+\mathbb{S}_{\ \circ\circ}^{\circ}\tilde{Z}^{i}|_{\circ},%
\end{array}%
\right.
\end{equation*}%
and
\begin{equation*}
\left\{
\begin{array}{cl}
Y_{\ \ |k|j}^{\circ}-Y_{\ \ |j|k}^{\circ} & =\mathbb{R}_{\ 0\ kj}^{\circ}Y^{c}+\mathbb{T}%
_{\ jk}^{\circ}Y^{\circ}|_{\circ}+\mathbb{T}_{\ jk}^{h}Y_{\ \ |h}^{\circ},\vspace*{2mm} \\
Y_{\ \ \ |k}^{\circ}|_{\circ}-Y^{\circ}|_{\circ |k} & =\mathbb{P}_{\ h\ k\circ}^{\circ}Y^{h}-\mathbb{P}%
_{\ k\circ}^{\circ}Y^{\circ}|_{\circ}-\mathbb{H}_{\circ k}^{\circ}Y^{\circ}|_{\circ},\vspace*{2mm} \\
Y^{\circ}|_{\circ}|_{\circ}-Y^{\circ}|_{\circ}|_{\circ} & =\mathbb{S}_{\ \circ\ \circ\circ}^{\circ}Y^{\circ}+\mathbb{S}%
_{\ \circ\circ}^{\circ}Y^{\circ}|_{\circ}.%
\end{array}%
\right.
\end{equation*}
We consider the following $1$-forms associated to distinguished linear $\left( \rho
,\eta \right) $-connection $\left( \left( \rho ,\eta \right) H,\left( \rho
,\eta \right) V\right) $%
\begin{equation*}
\begin{array}{c}
\left( \rho ,\eta \right) \omega _{\beta }^{\alpha }=\left( \rho ,\eta
\right) H_{\beta \gamma }^{\alpha }d\tilde{z}^{\gamma }+\left( \rho ,\eta
\right) V_{\beta \circ}^{\alpha }\delta \tilde{y}^{\circ},\vspace*{2mm} \\
\left( \rho ,\eta \right) \omega _{\circ}^{\circ}=\left( \rho ,\eta \right)
H_{\circ\gamma }^{\circ}d\tilde{z}^{\gamma }+\left( \rho ,\eta \right)
V_{\circ\circ}^{\circ}\delta \tilde{y}^{\circ}.
\end{array}%
\end{equation*}
Moreover, we obtain the following torsion $2$-forms
\begin{equation*}
\left\{
\begin{array}{cl}
\left( \rho ,\eta ,h\right) \mathbb{T}^{\alpha }\!\!\!\!& \displaystyle=\frac{1}{2}%
\left( \rho ,\eta ,h\right) \mathbb{T}_{~\beta \gamma }^{\alpha }d\tilde{z}%
^{\beta }\wedge d\tilde{z}^{\gamma }+\left( \rho ,\eta ,h\right) \mathbb{P}%
_{\ \beta \circ}^{\alpha }d\tilde{z}^{\beta }\wedge \delta \tilde{y}^{\circ},\vspace*{%
2mm} \\
\left( \rho ,\eta ,h\right) \mathbb{T}^{\circ}\!\!\!\!&=\displaystyle\frac{1}{2}\left(
\rho ,\eta ,h\right) \mathbb{T}_{~\beta \gamma }^{\circ}d\tilde{z}^{\beta
}\wedge d\tilde{z}^{\gamma }+\left( \rho ,\eta ,h\right) \mathbb{P}_{\ \beta
\circ}^{\circ}d\tilde{z}^{\beta }\wedge \delta \tilde{y}^{\circ}\vspace*{2mm} \\
& \ \ +\displaystyle\frac{1}{2}\left( \rho ,\eta ,h\right) \mathbb{S}%
_{\ \circ\circ}^{\circ}\delta \tilde{y}^{\circ}\wedge \delta \tilde{y}^{\circ},
\end{array}%
\right.
\end{equation*}%
and
\begin{equation*}
\left\{
\begin{array}{cl}
\left( \rho ,\eta ,h\right) \mathbb{R}_{\ \beta }^{\alpha }\!\!\!\!& =\displaystyle%
\frac{1}{2}\left( \rho ,\eta ,h\right) \mathbb{R}_{\ \beta\ \gamma \theta
}^{\alpha }d\tilde{z}^{\gamma }\wedge d\tilde{z}^{\theta }+\left( \rho ,\eta
,h\right) \mathbb{P}_{\ \beta \ \gamma \circ}^{\alpha }d\tilde{z}^{\gamma }\wedge
\delta \tilde{y}^{\circ}\vspace*{2mm} \\
&\ \ +\frac{1}{2}\left( \rho ,\eta ,h\right) \mathbb{S}_{\ \beta
\ \circ\circ}^{\alpha }\delta \tilde{y}^{\circ}\wedge \delta \tilde{y}^{\circ},\vspace*{2mm}
\\
\left( \rho ,\eta ,h\right) \mathbb{R}_{\ \circ}^{\circ}\!\!\!\!& \displaystyle=\frac{1}{2}%
\left( \rho ,\eta ,h\right) \mathbb{R}_{\ \circ\ \gamma \theta }^{\circ}d\tilde{z}%
^{\gamma }\wedge d\tilde{z}^{\theta }+\left( \rho ,\eta ,h\right) \mathbb{P}%
_{\ \circ\ \gamma \circ}^{\circ}d\tilde{z}^{\gamma }\wedge \delta \tilde{y}^{\circ}\vspace*{2mm%
} \\
&\ \ +\frac{1}{2}\left( \rho ,\eta ,h\right) \mathbb{S}%
_{\ \circ\ \circ\circ}^{\circ}\delta \tilde{y}^{\circ}\wedge \delta \tilde{y}^{\circ}.
\end{array}%
\right.
\end{equation*}
\begin{theorem}\label{AP3}
The following identities of Cartan type hold
\begin{equation}\label{C1}
\begin{array}{c}
\left( \rho ,\eta ,h\right) \mathbb{T}^{\alpha }=d^{\left( \rho ,\eta
\right) TE}\left( d\tilde{z}^{\alpha }\right) +\left( \rho ,\eta \right)
\omega _{\beta }^{\alpha }\wedge d\tilde{z}^{\beta },\vspace*{2mm} \\
\left( \rho ,\eta ,h\right) \mathbb{T}^{\circ}=d^{\left( \rho ,\eta \right)
TE}\left( \delta \tilde{y}^{\circ}\right) +\left( \rho ,\eta \right) \omega
_{\circ}^{\circ}\wedge \delta \tilde{y}^{\circ},%
\end{array}
\end{equation}
and%
\begin{equation}\label{C2}
\begin{array}{c}
\left( \rho ,\eta ,h\right) \mathbb{R}_{\ \beta }^{\alpha }=d^{\left( \rho
,\eta \right) TE}\left( \left( \rho ,\eta \right) \omega _{\beta }^{\alpha
}\right) +\left( \rho ,\eta \right) \omega _{\gamma }^{\alpha }\wedge \left(
\rho ,\eta \right) \omega _{\beta }^{\gamma },\vspace*{2mm} \\
\left( \rho ,\eta ,h\right) \mathbb{R}_{\ \circ}^{\circ}=d^{\left( \rho ,\eta \right)
TE}\left( \left( \rho ,\eta \right) \omega _{\circ}^{\circ}\right) +\left( \rho
,\eta \right) \omega _{\circ}^{\circ}\wedge \left( \rho ,\eta \right) \omega
_{\circ}^{\circ}.%
\end{array}
\end{equation}
\end{theorem}
In particular, if $\left( \rho ,\eta ,h\right) =\left(
Id_{TM},Id_{M},Id_{M}\right) $ and the Lie bracket $\left[ ,\right] _{TM}$
is the usual Lie bracket, then the identities of Cartan type (\ref{C1}) and (\ref{C2}) reduce to
\begin{equation*}
\begin{array}{c}
\mathbb{T}^{i}=d^{\left( Id_{TE},Id_{E}\right) TE}\left( d\tilde{z}%
^{i}\right)+\omega _{j}^{i}\wedge d\tilde{z}^{j},\vspace*{2mm} \\
\mathbb{T}^{\circ}=d^{\left( Id_{TE},Id_{E}\right) TE}\left( \delta \tilde{y}%
^{\circ}\right) +\omega_{\circ}^{\circ}\wedge \delta \tilde{y}^{\circ},%
\end{array}
\end{equation*}%
and
\begin{equation*}
\begin{array}{c}
\mathbb{R}_{~j}^{i}=d^{\left( Id_{TE},Id_{E}\right) TE}\left( \omega
_{j}^{i}\right) +\omega _{k}^{i}\wedge \omega _{j}^{k},\vspace*{2mm} \\
\mathbb{R}_{\ \circ}^{\circ}=d^{\left( Id_{TE},Id_{E}\right) TE}\left( \omega
_{\circ}^{\circ}\right) +\omega _{\circ}^{\circ}\wedge \omega _{\circ}^{\circ}.%
\end{array}
\end{equation*}
\begin{remark}\label{Ap4}
For any $X,Y,Z\in \Gamma \left(
\left( \rho ,\eta \right) TE,\left( \rho ,\eta \right) \tau _{E},E\right) $,
the following identities hold
\begin{equation*}
\begin{array}{rc}
\mathcal{V}\left( \rho ,\eta ,h\right) \mathbb{R}\left( X,Y\right) \mathcal{H%
}Z & =0,\vspace*{2mm} \\
\mathcal{H}\left( \rho ,\eta ,h\right) \mathbb{R}\left( X,Y\right) \mathcal{V%
}Z & =0,%
\end{array}
\end{equation*}%
\begin{equation*}
\begin{array}{cc}
\mathcal{V}D_{X}\left( \left( \rho ,\eta ,h\right) \mathbb{R}\left(
Y,Z\right) \mathcal{H}U\right) & =0,\vspace*{2mm} \\
\mathcal{H}D_{X}\left( \left( \rho ,\eta ,h\right) \mathbb{R}\left(
Y,Z\right) \mathcal{V}U\right) & =0,%
\end{array}
\end{equation*}%
and%
\begin{equation*}
\begin{array}{cl}
\left( \rho ,\eta ,h\right) \mathbb{R}\left( X,Y\right) Z=\mathcal{H}\left(
\rho ,\eta ,h\right) \mathbb{R}\left( X,Y\right) \mathcal{H}Z+\mathcal{V}%
\left( \rho ,\eta ,h\right) \mathbb{R}\left( X,Y\right) \mathcal{V}Z. &
\end{array}
\end{equation*}
\end{remark}
Using the formulas of Bianchi type, Theorem \ref{AP3} and Remark \ref{Ap4}, we
obtain the following
\begin{theorem}
The identities of Bianchi type:
\begin{equation*}
\left\{
\begin{array}{l}
\underset{cyclic\left( X,Y,Z\right) }{\dsum }\left\{ \mathcal{H}\left( \rho
,\eta \right) D_{X}\left( \left( \rho ,\eta ,h\right) \mathbb{T}\left(
Y,Z\right) \right) -\mathcal{H}\left( \rho ,\eta ,h\right) \mathbb{R}\left(
X,Y\right) Z\right. \vspace*{2mm} \\
\qquad \qquad +\mathcal{H}\left( \rho ,\eta ,h\right) \mathbb{T}\left(
\mathcal{H}\left( \rho ,\eta ,h\right) \mathbb{T}\left( X,Y\right) ,Z\right)
\vspace*{2mm} \\
\qquad \qquad \left. +\mathcal{H}\left( \rho ,\eta ,h\right) \mathbb{T}%
\left( \mathcal{V}\left( \rho ,\eta ,h\right) \mathbb{T}\left( X,Y\right)
,Z\right) \right\} =0,\vspace*{4mm} \\
\underset{cyclic\left( X,Y,Z\right) }{\dsum }\left\{ \mathcal{V}\left( \rho
,\eta \right) D_{X}\left( \left( \rho ,\eta ,h\right) \mathbb{T}\left(
Y,Z\right) \right) -\mathcal{V}\left( \rho ,\eta ,h\right) \mathbb{R}\left(
X,Y\right) Z\right. \vspace*{2mm} \\
\qquad \qquad +\mathcal{V}\left( \rho ,\eta ,h\right) \mathbb{T}\left(
\mathcal{H}\left( \rho ,\eta ,h\right) \mathbb{T}\left( X,Y\right) ,Z\right)
\vspace*{2mm} \\
\qquad \qquad \left. +\mathcal{V}\left( \rho ,\eta ,h\right) \mathbb{T}%
\left( \mathcal{V}\left( \rho ,\eta ,h\right) \mathbb{T}\left( X,Y\right)
,Z\right) \right\} =0.%
\end{array}%
\right.
\end{equation*}%
and
\begin{equation*}
\left\{
\begin{array}{l}
\underset{cyclic\left( X,Y,Z,U\right) }{\dsum }\left\{ \mathcal{H}\left(
\rho ,\eta \right) D_{X}\left( \left( \rho ,\eta ,h\right) \mathbb{R}\left(
Y,Z\right) U\right) \right. \vspace*{2mm} \\
\qquad \qquad -\mathcal{H}\left( \rho ,\eta ,h\right) \mathbb{R}\left(
\mathcal{H}\left( \rho ,\eta ,h\right) \mathbb{T}\left( X,Y\right) ,Z\right)
U\vspace*{2mm} \\
\qquad \qquad \left. -\mathcal{H}\left( \rho ,\eta ,h\right) \mathbb{R}%
\left( \mathcal{V}\left( \rho ,\eta ,h\right) \mathbb{T}\left( X,Y\right)
,Z\right) U\right\} =0,\vspace*{4mm} \\
\underset{cyclic\left( X,Y,Z,U\right) }{\dsum }\left\{ \mathcal{V}\left(
\rho ,\eta \right) D_{X}\left( \left( \rho ,\eta ,h\right) \mathbb{R}\left(
Y,Z\right) U\right) \right. \vspace*{2mm} \\
\qquad \qquad -\mathcal{V}\left( \rho ,\eta ,h\right) \mathbb{R}\left(
\mathcal{H}\left( \rho ,\eta ,h\right) \mathbb{T}\left( X,Y\right) ,Z\right)
U\vspace*{2mm} \\
\qquad \qquad \left. -\mathcal{V}\left( \rho ,\eta ,h\right) \mathbb{R}%
\left( \mathcal{V}\left( \rho ,\eta ,h\right) \mathbb{T}\left( X,Y\right)
,Z\right) U\right\} =0,%
\end{array}%
\right.
\end{equation*}%
hold for any $X,Y,Z\in \Gamma \left( \left( \rho ,\eta \right)
TE,\left( \rho ,\eta \right) \tau _{E},E\right) .$
\end{theorem}
\begin{corollary}
Using the following sections $\left( \delta _{\theta },\delta _{\gamma },\delta _{\beta }\right) $, we obtain:
\begin{equation*}
\left\{
\begin{array}{l}
\underset{cyclic\left( \beta ,\gamma ,\theta \right) }{\dsum }\left\{ \left(
\rho ,\eta ,h\right) \mathbb{T}_{\ \ \beta \gamma _{|\theta }}^{\alpha
}-\left( \rho ,\eta ,h\right) \mathbb{R}_{\ \beta\ \gamma \theta }^{\alpha
}\right. \vspace*{2mm} \\
\qquad \qquad \left. +\left( \rho ,\eta ,h\right) \mathbb{T}_{\ \gamma \theta
}^{\lambda }\left( \rho ,\eta ,h\right) \mathbb{T}_{\ \beta \gamma }^{\alpha
}+\left( \rho ,\eta ,h\right) \mathbb{T}_{\ \gamma \theta }^{\circ}\left( \rho
,\eta ,h\right) \mathbb{T}_{\ \beta \circ}^{\alpha }\right\} =0,\vspace*{4mm} \\
\underset{cyclic\left( \beta ,\gamma ,\theta \right) }{\dsum }\left\{ \left(
\rho ,\eta ,h\right) \mathbb{T}_{\ \ \beta \gamma _{|\theta }}^{\circ}+\left(
\rho ,\eta ,h\right) \mathbb{T}_{\ \gamma \theta }^{\alpha }\left( \rho ,\eta
,h\right) \mathbb{P}_{\ \beta \alpha }^{\circ}\right. \vspace*{2mm} \\
\qquad \qquad \left. +\left( \rho ,\eta ,h\right) \mathbb{P}_{\ \gamma \theta
}^{\circ}\left( \rho ,\eta ,h\right) \mathbb{P}_{\ \circ\beta }^{\circ}\right\} =0,%
\end{array}%
\right.
\end{equation*}%
and 
\begin{equation*}
\left\{
\begin{array}{l}
\underset{cyclic\left( \beta ,\gamma ,\theta ,\lambda \right) }{\dsum }%
\left\{ \left( \rho ,\eta ,h\right) \mathbb{R}_{\ \ \beta \ \gamma \theta
_{|\lambda }}^{\alpha }-\left( \rho ,\eta ,h\right) \mathbb{T}_{~\theta
\lambda }^{\mu }\left( \rho ,\eta ,h\right) \mathbb{R}_{\ \beta \ \gamma \mu
}^{\alpha }\right. \vspace*{2mm} \\
\qquad \qquad \left. -\left( \rho ,\eta ,h\right) \mathbb{T}_{\ \theta
\lambda }^{\circ}\left( \rho ,\eta ,h\right) \mathbb{P}_{\ \beta \ \gamma
\circ}^{\alpha }\right\} =0,\vspace*{4mm} \\
\underset{cyclic\left( \beta ,\gamma ,\theta ,\lambda \right) }{\dsum }%
\left\{ \left( \rho ,\eta ,h\right) \mathbb{R}_{\ \ \circ\ \gamma \theta
_{|\lambda }}^{\circ}-\left( \rho ,\eta ,h\right) \mathbb{T}_{\ \theta \lambda
}^{\mu }\left( \rho ,\eta ,h\right) \mathbb{R}_{\ \circ\ \gamma \mu }^{\circ}\right.
\vspace*{2mm} \\
\qquad \qquad \left. -\left( \rho ,\eta ,h\right) \mathbb{T}_{\ \theta
\lambda }^{\circ}\left( \rho ,\eta ,h\right) \mathbb{P}_{\ \circ\ \gamma
\circ}^{\circ}\right\} =0.%
\end{array}%
\right.
\end{equation*}
\end{corollary}
Note that using another base of sections, we can obtain new identities of Bianchi
type necessary in the applications.
\section{(Pseudo) generalized Kaluza-Klein $G$-spaces}
Let $(\rho, \eta)\Gamma$ be a $(\rho, \eta)$-connection for the Kaluza-Klein bundle $(E, \pi, M)$ and let $((\rho, \eta)H, (\rho, \eta)V)$ be a distinguished linear $(\rho, \eta)$-connection for the Lie algebroid generalized tangent bundle
\[
\Big(((\rho, \eta)TE, (\rho, \eta)\tau_E, E), [, ]_{(\rho, \eta)TE}, (\tilde{\rho}, Id_E)\Big).
\]
We consider the following tensor field
\begin{equation*}
G=g_{\alpha \beta }d\tilde{z}^{\alpha }\otimes d\tilde{z}^{\beta
}+g_{\circ\circ}\delta \tilde{y}^{\circ}\otimes \delta \tilde{y}^{\circ}\in \mathcal{T}%
_{22}^{00}\left( \left( \rho ,\eta \right) TE,\left( \rho ,\eta \right) \tau
_{E},E\right).
\end{equation*}
If $G$ is symmetric and for any point $u_{x}\in E$, matrice $\left(g_{\alpha \beta }\left( u_{x}\right)\right)$ is nondegenerate and also $g_{\circ\circ}(u_x)$ is invertible, then $G$ is called \emph{pseudometrical structure}.
Moreover, if  $\left( g_{\alpha \beta }\left( u_{x}\right)
\right)$ and $g_{\circ\circ}(u_{x})$ have
constant signature, then  $G$ is called \emph{%
metrical structure}\textit{.}

If $\alpha ,\beta \in 1, \ldots, p$, then for any local fiber chart
$\left( U,s_{U}\right) $ of $\left( E,\pi ,M\right) $, we consider the real
functions
\begin{equation*}
\begin{array}{ccc}
\pi ^{-1}\left( U\right) & ^{\underrightarrow{~\ \ \tilde{g}^{\beta \alpha
}~\ \ }} & \!\!\!\!\!\mathbb{R},
\end{array}%
\end{equation*}%
and
\begin{equation*}
\begin{array}{ccc}
\pi ^{-1}\left( U\right) & ^{\underrightarrow{~\ \ \tilde{g}^{\circ\circ}~\ \ }} \!\!\!\!\!&
\mathbb{R},
\end{array}%
\end{equation*}%
such that%
\begin{equation*}
\begin{array}{c}
\left(\tilde{g}^{\beta \alpha }\left( u_{x}\right) \right)
=\left(g_{\alpha \beta }\left( u_{x}\right) \right)^{-1},
\end{array}%
\end{equation*}%
and
\begin{equation*}
\begin{array}{c}
\tilde{g}^{\circ\circ}(u_{x})=g_{\circ\circ}(u_{x})^{-1},%
\end{array}%
\end{equation*}%
for any $u_{x}\in \pi ^{-1}\left( U\right) \backslash \left\{ 0_{x}\right\}$.

If around each point $x\in M$ there exists a local
vector chart $\left( U,s_{U}\right) $ and a local chart $\left(
U,\xi _{U}\right) $ such that $g_{\alpha \beta }\circ s_{U}^{-1}\circ \left(
\xi _{U}^{-1}\times Id_{\mathbb{R}}\right) \left( x,y\right) $ (respectively $g_{\circ\circ}\circ s_{U}^{-1}\circ \left( \xi _{U}^{-1}\times Id_{\mathbb{R}}\right) \left( x,y\right)$) depends only on $x$, for any $u_{x}\in \pi
^{-1}\left( U\right) ,$ then we say that the (pseudo)metrical
structure $G$ is is a {\it Riemannian $\mathcal{H}$-(pseudo) metrical structure} (respectively, {\it Riemannian $\mathcal{V}$-(pseudo)metrical structure}). Moreover, $G$ is called {\it Riemannian (pseudo) metrical structure} if it is $\mathcal{H}$- and $\mathcal{V}$-(pseudo) metrical structures.
\begin{definition}
If there exists a (pseudo) metrical structure%
\begin{equation}\label{metr}
G=g_{\alpha \beta }d\tilde{z}^{\alpha }\otimes d\tilde{z}^{\beta
}+g_{\circ\circ}\delta \tilde{y}^{\circ}\otimes \delta \tilde{y}^{\circ},
\end{equation}
such that
\begin{equation}\label{AP5}
\begin{array}{c}
\left( \rho ,\eta \right) D_{X}G=0,~\forall X\in \Gamma \left( \left( \rho
,\eta \right) TE,\left( \rho ,\eta \right) \tau _{E},E\right),
\end{array}
\end{equation}
then the Lie algebroid generalized tangent bundle
\begin{equation*}
\begin{array}{c}
\left( \left( \left( \rho ,\eta \right) TE,\left( \rho ,\eta \right) \tau
_{E},E\right) ,\left[ ,\right] _{\left( \rho ,\eta \right) TE},\left( \tilde{%
\rho},Id_{E}\right) \right),
\end{array}%
\end{equation*}%
is called (pseudo) generalized Kaluza-Klein $G$-space. Moreover, the local real functions $(\rho, \eta)\Gamma^\circ_\alpha$ and $(\rho, \eta)R^\circ_{\,\ \alpha\beta}$ are called the electromagnetic potentials and the components of the electromagnetic field, respectively.
\end{definition}
One can deduce that the condition (\ref{AP5}) is equivalent with the following:
\begin{equation*}
\begin{array}{c}
g_{\alpha \beta \mid \gamma}=0,\ \ g_{\circ\circ\mid \gamma }=0,\ \ g_{\alpha \beta
}\mid _{\circ}=0\,,\,\,g_{\circ\circ}\mid _{\circ}=0.%
\end{array}
\end{equation*}
If $g_{\alpha \beta \mid \gamma }{=}0$ and $\,g_{\circ\circ\mid \gamma }{=}0$, then
we will say that the Lie algebroid generalized tangent bundle 
\begin{equation*}
\begin{array}{c}
\left( \left( \left( \rho ,\eta \right) TE,\left( \rho ,\eta \right) \tau
_{E},E\right) ,\left[ ,\right] _{\left( \rho ,\eta \right) TE},\left( \tilde{%
\rho},Id_{E}\right) \right),
\end{array}%
\end{equation*}
is \textit{$\mathcal{H}$-(pseudo) generalized Kaluza-Klein $G$-space}. Also, if $g_{\alpha \beta }|_{\circ}{=}0$ and $\,g_{\circ\circ}|_{\circ}{=}0$, then we say
that the Lie algebroid generalized tangent bundle
\begin{equation*}
\begin{array}{c}
\left( \left( \left( \rho ,\eta \right) TE,\left( \rho ,\eta \right) \tau
_{E},E\right) ,\left[ ,\right] _{\left( \rho ,\eta \right) TE},\left( \tilde{%
\rho},Id_{E}\right) \right),
\end{array}%
\end{equation*}
is \textit{$\mathcal{V}$-(pseudo) generalized Kaluza-Klein $G$-space}.
\begin{theorem}\label{Th}
If $\left( \left( \rho ,\eta \right) \mathring{H}%
,\left( \rho ,\eta \right) \mathring{V}\right) $ is a distinguished
linear $\left( \rho ,\eta \right) $-connection for the Lie algebroid generalized tangent bundle
\begin{equation*}
\begin{array}{c}
\left( \left( \left( \rho ,\eta \right) TE,\left( \rho ,\eta \right) \tau
_{E},E\right) ,\left[ ,\right] _{\left( \rho ,\eta \right) TE},\left( \tilde{%
\rho},Id_{E}\right)\right),
\end{array}%
\end{equation*}
and $G$ is a (pseudo) metrical structure given by (\ref{metr}), then the real local functions:
\begin{equation*}
\begin{array}{ll}
\left( \rho ,\eta \right) H_{\beta \gamma }^{\alpha }\!\! & =\displaystyle%
\frac{1}{2}\tilde{g}^{\alpha \varepsilon }\left( \Gamma \left( \tilde{\rho}%
,Id_{E}\right) \left( \tilde{\delta}_{\gamma }\right) g_{\varepsilon \beta
}+\Gamma \left( \tilde{\rho},Id_{E}\right) \left( \tilde{\delta}_{\beta
}\right) g_{\varepsilon \gamma }-\Gamma \left( \tilde{\rho},Id_{E}\right)
\left( \tilde{\delta}_{\varepsilon }\right) g_{\beta \gamma }\right.
\vspace*{1mm} \\
& \left. +g_{\theta \varepsilon }L_{\gamma \beta }^{\theta }\circ h\circ \pi
-g_{\beta \theta }L_{\gamma \varepsilon }^{\theta }\circ h\circ \pi
-g_{\theta \gamma }L_{\beta \varepsilon }^{\theta }\circ h\circ \pi \right),
\vspace*{2mm} \\
\left( \rho ,\eta \right) H_{\circ\gamma }^{\circ}\!\! & =\left( \rho ,\eta \right)
\mathring{H}_{\circ\gamma }^{\circ}+\displaystyle\frac{1}{2}\tilde{g}^{\circ\circ}g_{\circ\circ%
\overset{0}{\mid }\gamma },\vspace*{2mm} \\
\left( \rho ,\eta \right) V_{\beta \circ}^{\alpha }\!\! & =\left( \rho ,\eta
\right) \mathring{V}_{\beta \circ}^{\alpha }+\displaystyle\frac{1}{2}\tilde{g}%
^{\alpha \varepsilon }g_{\beta \varepsilon \overset{\circ}{\mid }\circ},\vspace*{2mm}
\\
\left( \rho ,\eta \right) V_{\circ\circ}^{\circ}\!\! & =\displaystyle\frac{1}{2}\tilde{g}%
^{\circ\circ}\left( \Gamma \left( \tilde{\rho},Id_{E}\right) \left( \dot{%
\tilde{\partial}}_{\circ}\right) g_{\circ\circ}+\Gamma \left( \tilde{\rho},Id_{E}\right)
\left( \dot{\tilde{\partial}}_{\circ}\right) g_{\circ\circ}-\Gamma \left(
\tilde{\rho},Id_{E}\right) \left( \dot{\tilde{\partial}}%
_{\circ}\right) g_{\circ\circ}\right),
\end{array}%
\end{equation*}
are the components of a distinguished linear $(\rho, \eta)$-connection such that the Lie algebroid generalized tangent bundle
\[
\Big(((\rho, \eta)TE, (\rho, \eta)\tau_E, E), [, ]_{(\rho, \eta)TE}, (\tilde{\rho}, Id_E)\Big),
\]
becomes (pseudo) generalized Kaluza-Klein G-space.
\end{theorem}
\begin{corollary}
In the particular case of Lie algebroids, $%
\left( \eta ,h\right) =\left( Id_{M},Id_{M}\right)$, then we obtain
\begin{equation*}
\begin{array}{ll}
\rho H_{\beta \gamma }^{\alpha }\!\! & =\displaystyle\frac{1}{2}\tilde{g}%
^{\alpha \varepsilon }\left( \Gamma \left( \tilde{\rho},Id_{E}\right) \left(
\tilde{\delta}_{\gamma }\right) g_{\varepsilon \beta }+\Gamma \left( \tilde{%
\rho},Id_{E}\right) \left( \tilde{\delta}_{\beta }\right) g_{\varepsilon
\gamma }-\Gamma \left( \tilde{\rho},Id_{E}\right) \left( \tilde{\delta}%
_{\varepsilon }\right) g_{\beta \gamma }\right. \vspace*{1mm} \\
& \left. +g_{\theta \varepsilon }L_{\gamma \beta }^{\theta }\circ \pi
-g_{\beta \theta }L_{\gamma \varepsilon }^{\theta }\circ \pi -g_{\theta
\gamma }L_{\beta \varepsilon }^{\theta }\circ \pi \right) ,\vspace*{2mm} \\
\rho H_{\circ\gamma }^{\circ}\!\! & =\rho \mathring{H}_{\circ\gamma }^{\circ}+\displaystyle%
\frac{1}{2}\tilde{g}^{\circ\circ}g_{\circ\circ\overset{\circ}{\mid }\gamma },\vspace*{2mm} \\
\rho V_{\beta \circ}^{\alpha }\!\! & =\rho \mathring{V}_{\beta \circ}^{\alpha }+%
\displaystyle\frac{1}{2}\tilde{g}^{\alpha \varepsilon }g_{\beta \varepsilon
\overset{\circ}{\mid }\circ},\vspace*{2mm} \\
\rho V_{\circ\circ}^{\circ}\!\! & =\displaystyle\frac{1}{2}\tilde{g}^{\circ\circ}\left( \Gamma
\left( \tilde{\rho},Id_{E}\right) \left( \dot{\tilde{\partial}}%
_{\circ}\right) g_{\circ\circ}+\Gamma \left( \tilde{\rho},Id_{E}\right) \left( \dot{\tilde{\partial}}_{\circ}\right) g_{\circ\circ}-\Gamma \left( \tilde{\rho}%
,Id_{E}\right) \left( \dot{\tilde{\partial}}_{\circ}\right)
g_{\circ\circ}\right) \vspace*{1mm}%
\end{array}%
\end{equation*}
\end{corollary}
In the classical case, $\left( \rho ,\eta ,h\right) =\left(
Id_{TM},Id_{M},Id_{M}\right)$, we obtain
\begin{equation*}
\begin{array}{ll}
H_{jk}^{i}\!\! & =\displaystyle\frac{1}{2}\tilde{g}^{ih}\left( \delta
_{k}g_{hj}+\delta _{j}g_{hk}-\delta _{h}g_{jk}\vspace*{1mm}\right),\\
H_{\circ k}^{\circ}\!\! & =\mathring{H}_{\circ k}^{\circ}+\displaystyle\frac{1}{2}\tilde{g}%
^{\circ\circ}g_{\circ\circ\overset{\circ}{\mid }k},\vspace*{2mm} \\
V_{j\circ}^{i}\!\! & =\mathring{V}_{j\circ}^{i}+\displaystyle\frac{1}{2}\tilde{g}%
^{ih}g_{jh\overset{\circ}{\mid }\circ},\vspace*{2mm} \\
V_{\circ\circ}^{\circ}\!\! & =\displaystyle\frac{1}{2}\tilde{g}^{\circ\circ}\left( \dot{\partial}%
_{\circ}g_{_{\circ\circ}}+\dot{\partial}_{\circ}g_{_{\circ\circ}}-\dot{\partial}_{\circ}g_{_{\circ\circ}}\right).
\end{array}
\end{equation*}
\begin{theorem}
If the distinguished linear $%
\left( \rho ,\eta \right) $-connection $\left( \left( \rho ,\eta
\right) \mathring{H},\left( \rho ,\eta \right) \mathring{V}\right) $ coincides with the Berwald linear $\left( \rho ,\eta \right) $-connection in Theorem \ref{Th}, then the local real functions:
\begin{equation*}
\begin{array}{ll}
\left( \rho ,\eta \right) \overset{c}{H}_{\beta \gamma }^{\alpha }\!\!\! & =%
\displaystyle\frac{1}{2}\tilde{g}^{\alpha \varepsilon }\left( \Gamma \left(
\tilde{\rho},Id_{E}\right) \left( \tilde{\delta}_{\gamma }\right)
g_{\varepsilon \beta }+\Gamma \left( \tilde{\rho},Id_{E}\right) \left(
\tilde{\delta}_{\beta }\right) g_{\varepsilon \gamma }\right. \vspace*{1mm}
\\
& -\Gamma \left( \tilde{\rho},Id_{E}\right) \left( \tilde{\delta}%
_{\varepsilon }\right) g_{\beta \gamma }+g_{\theta \varepsilon }L_{\gamma
\beta }^{\theta }\circ h\circ \pi \vspace*{1mm}\\
&\left. -g_{\beta
\theta }L_{\gamma \varepsilon }^{\theta }\circ h\circ \pi-g_{\theta \gamma
}L_{\beta \varepsilon }^{\theta }\circ h\circ \pi \right),\vspace*{2mm}\\
\left( \rho ,\eta \right) \overset{c}{H}_{\circ\gamma}^{\circ}\!\!\! & =%
\displaystyle\dot{\partial}_\circ\left(\left( \rho ,\eta \right) \Gamma _{\gamma }^{\circ}\right)+\frac{1}{2}\tilde{g}^{\circ\circ}g_{\circ\circ\overset{\circ}{\mid }\gamma },%
\vspace*{2mm} \\
\left( \rho ,\eta \right) \overset{c}{V}_{\beta \circ}^{\alpha }\!\! & =%
\displaystyle\frac{1}{2}\tilde{g}^{\alpha \varepsilon }\dot{\partial}_\circ
g_{\beta \varepsilon },\vspace*{2mm} \\
\left( \rho ,\eta \right) \overset{c}{V}_{\circ\circ}^{\circ}\!\!\! & =\displaystyle%
\frac{1}{2}\tilde{g}^{\circ\circ}\dot{\partial}_\circ g_{\circ\circ},
\end{array}
\end{equation*}%
are the components of a distinguished linear $\left( \rho ,\eta
\right) $-connection such that the Lie algebroid generalized tangent
bundle
\begin{equation*}
\begin{array}{c}
\left( \left( \left( \rho ,\eta \right) TE,\left( \rho ,\eta \right) \tau
_{E},E\right) ,\left[ ,\right] _{\left( \rho ,\eta \right) TE},\left( \tilde{%
\rho},Id_{E}\right) \right),
\end{array}%
\end{equation*}%
becomes (pseudo) generalized Kaluza-Klein $G$-space.
\end{theorem}
Moreover, if the (pseudo) metrical structure $G$ is $\mathcal{H%
}$- and $\mathcal{V}$-Riemannian, then the local real
functions:
\begin{equation*}
\begin{array}{cl}
(\rho ,\eta )\overset{c}{H}_{\beta \gamma }^{\alpha } & {=}\frac{1}{2}\tilde{%
g}^{\alpha \varepsilon }\left( \rho _{\gamma }^{k}{\circ }h{\circ }\pi \frac{%
\partial g_{\varepsilon \beta }}{\partial x^{k}}+\rho _{\beta }^{j}{\circ }h{%
\circ }\pi \frac{\partial g_{\varepsilon \gamma }}{\partial x^{j}}-\rho
_{\varepsilon }^{e}{\circ }h{\circ }\pi \frac{\partial g_{\beta \gamma }}{%
\partial x^{e}}+\right.  \\
& \left. +g_{\theta \varepsilon }L_{\gamma \beta }^{\theta }{\circ }h{\circ }%
\pi -g_{\beta \theta }L_{\gamma \varepsilon }^{\theta }{\circ }h{\circ }\pi
-g_{\theta \gamma }L_{\beta \varepsilon }^{\theta }{\circ }h{\circ }\pi
\right) ,\vspace*{1mm} \\
\left( \rho ,\eta \right) \overset{c}{H}_{\circ\gamma }^{\circ} & {=}\dot{\partial}_\circ
\left(\left( \rho ,\eta \right) \Gamma _{\gamma }^{\circ}\right)+\frac{1}{2}%
\tilde{g}^{\circ\circ}\left( \rho _{\gamma }^{i}{\circ }h{\circ }\pi \frac{\partial
g_{_{\circ\circ}}}{\partial x^{i}}-2\dot{\partial}_\circ
\left(\left( \rho ,\eta \right) \Gamma _{\gamma }^{\circ}\right)\right) , \\
\left( \rho ,\eta \right) \overset{c}{V}_{\beta \circ}^{\alpha } & =0, \\
\left( \rho ,\eta \right) \overset{c}{V}_{\circ\circ}^{\circ} & =0.%
\end{array}
\end{equation*}
are the components of a distinguished linear $\left( \rho ,\eta
\right) $-connection such that the Lie algebroid generalized tangent
bundle becomes (pseudo) generalized Kaluza-Klein $G$-space.
\begin{corollary}
In the particular case of Lie algebroids,
$\left( \eta ,h\right) =\left( Id_{M},Id_{M}\right)$, then we obtain
\begin{equation*}
\begin{array}{ll}
\rho \overset{c}{H}_{\beta \gamma }^{\alpha }\!\!\! & =\displaystyle\frac{1}{%
2}\tilde{g}^{\alpha \varepsilon }\left( \Gamma \left( \tilde{\rho}%
,Id_{E}\right) \left( \tilde{\delta}_{\gamma }\right) g_{\varepsilon \beta
}+\Gamma \left( \tilde{\rho},Id_{E}\right) \left( \tilde{\delta}_{\beta
}\right) g_{\varepsilon \gamma }\right. \vspace*{1mm}\\
&\ \ -\Gamma \left( \tilde{\rho},Id_{E}\right) \left( \tilde{\delta}%
_{\varepsilon }\right) g_{\beta \gamma }+g_{\theta \varepsilon }L_{\gamma
\beta }^{\theta }\circ \pi \vspace*{1mm}\left. -g_{\beta \theta }L_{\gamma
\varepsilon }^{\theta }\circ \pi -g_{\theta \gamma }L_{\beta \varepsilon
}^{\theta }\circ \pi \right),\\
\rho \overset{c}{H}_{\circ\gamma }^{\circ}\!\!\! & =\displaystyle\dot{\partial}_\circ(\rho
\Gamma _{\gamma }^{\circ})+\frac{1}{2}\tilde{g}^{\circ\circ}g_{\circ\circ\overset%
{\circ}{\mid }\gamma },\vspace*{2mm} \\
\rho \overset{c}{V}_{\beta \circ}^{\alpha }\!\! & =\displaystyle\frac{1}{2}%
\tilde{g}^{\alpha \varepsilon }\dot{\partial}_\circ g_{\beta \varepsilon },\vspace*{2mm} \\
\rho \overset{c}{V}_{\circ\circ}^{\circ}\!\!\! & =\displaystyle\frac{1}{2}\tilde{g}%
^{\circ\circ}\dot{\partial}_\circ g_{\circ\circ}.
\end{array}
\end{equation*}
If the (pseudo) metrical structure $G$ is $\mathcal{H}$-and $\mathcal{V}$-Riemannian, then
\begin{equation*}
\begin{array}{l}
\rho \overset{c}{H}_{\beta \gamma }^{\alpha }{=}\displaystyle\frac{1}{2}%
\tilde{g}^{\alpha \varepsilon }\left( \rho _{\gamma }^{k}{\circ }\pi \frac{%
\partial g_{\varepsilon \beta }}{\partial x^{k}}+\rho _{\beta }^{j}{\circ }%
\pi \frac{\partial g_{\varepsilon \gamma }}{\partial x^{j}}-\rho
_{\varepsilon }^{e}{\circ }\pi \frac{\partial g_{\beta \gamma }}{\partial
x^{e}}+\right. \vspace*{1mm} \\
\hfill \left. +g_{\theta \varepsilon }L_{\gamma \beta }^{\theta }{\circ }\pi
-g_{\beta \theta }L_{\gamma \varepsilon }^{\theta }{\circ }\pi -g_{\theta
\gamma }L_{\beta \varepsilon }^{\theta }{\circ }\pi \right) ,\vspace*{1mm}
\\
\rho \overset{c}{H}_{\circ\gamma }^{\circ}{=}\dot{\partial}_\circ(\rho \Gamma _{\gamma
}^{\circ})+\frac{1}{2}\tilde{g}^{\circ\circ}\rho _{\gamma }^{i}{\circ }%
\pi \frac{\partial g_{_{\circ\circ}}}{\partial x^{i}},\vspace*{2mm} \\
\rho \overset{c}{V}_{\beta \circ}^{\alpha }=0,\ \rho \overset{c}{V}_{00}^{0}=0.
\end{array}%
\end{equation*}
\end{corollary}
In the classical case, $\left( \rho ,\eta ,h\right) =\left(
Id_{TM},Id_{M},Id_{M}\right)$, we obtain
\begin{equation*}
\begin{array}{ll}
\overset{c}{H}_{jk}^{i}\!\!\! & =\displaystyle\frac{1}{2}\tilde{g}%
^{ih}\left( \delta _{k}g_{hj}+\delta _{j}g_{hk}-\delta _{h}g_{jk}\right),\\
\overset{c}{H}_{\circ k}^{\circ}\!\!\! & =\displaystyle\dot{\partial}_\circ\Gamma _{k}^{\circ}+\frac{1}{2}\tilde{g}^{\circ\circ}g_{\circ\circ\overset{\circ}{\mid }k},%
\vspace*{2mm} \\
\overset{c}{V}_{j\circ}^{i}\!\! & =\displaystyle\frac{1}{2}\tilde{g}^{ih}\dot{\partial}_\circ g_{jh},\vspace*{2mm} \\
\overset{c}{V}_{\circ\circ}^{\circ}\!\!\! & =\displaystyle\frac{1}{2}\tilde{g}%
^{\circ\circ}\dot{\partial}_\circ g_{_{\circ\circ}}.
\end{array}%
\end{equation*}
If the (pseudo)metrical structure $G$ is $\mathcal{H}$- and $\mathcal{V}$-Riemannian, then
\begin{equation*}
\begin{array}{l}
\overset{c}{H}_{jk}^{i}{=}\displaystyle\frac{1}{2}\tilde{g}^{ih}\left( \frac{%
\partial g_{hj}}{\partial x^{k}}+\frac{\partial g_{hk}}{\partial x^{j}}-%
\frac{\partial g_{jk}}{\partial x^{h}}\right),\\
\hfill  \\
\overset{c}{H}_{\circ k}^{\circ}{=}\dot{\partial}_\circ \Gamma _{k}^{\circ}+%
\frac{1}{2}\tilde{g}^{\circ\circ}\frac{\partial g_{_{\circ\circ}}}{\partial x^{k}},\vspace*{%
2mm} \\
\overset{c}{V}_{j\circ}^{i}=0,\ \overset{c}{V}_{\circ\circ}^{\circ}=0.%
\end{array}%
\end{equation*}
Let
\begin{equation*}
\begin{array}{c}
\left( \left( \left( \rho ,\eta \right) TE,\left( \rho ,\eta \right) \tau
_{E},E\right) ,\left[ ,\right] _{\left( \rho ,\eta \right) TE},\left( \tilde{%
\rho},Id_{E}\right) \right),
\end{array}%
\end{equation*}%
be a (pseudo) generalized Kaluza-Klein $G$-space.
\begin{definition}
If $\left( \rho ,\eta ,h\right)
\mathbb{R}_{\ \alpha \ \beta }$ and $\left( \rho ,\eta ,h\right) \mathbb{S}%
_{\ \circ\ \circ}$ are the components of Ricci tensor associated to distinguished
linear $\left( \rho ,\eta \right) $-connection
\[
\begin{array}{c}
\left( \left( \rho ,\eta \right) H,\left( \rho ,\eta \right) V\right) ,%
\end{array}%
\]%
then the scalar
\[
\begin{array}{c}
\left( \rho ,\eta ,h\right) \mathbb{R=}\left( \rho ,\eta ,h\right) \mathbb{R}%
_{~\alpha ~\beta }\tilde{g}^{\alpha \beta }+\left( \rho ,\eta ,h\right)
\mathbb{S}_{\ \circ\ \circ}\tilde{g}^{\circ\circ},
\end{array}
\]%
is called the scalar curvature of the distinguished linear $(\rho
,\eta )$-connection \\$((\rho ,\eta )H,(\rho ,\eta )V)$.
\end{definition}
\begin{definition}
The tensor field
\[
\begin{array}{cc}
\left( \rho ,\eta ,h\right) \mathbb{T} & =\left( \rho ,\eta ,h\right)
\mathbb{T}_{\ \alpha \ \beta }d\tilde{z}^{\alpha }\otimes d\tilde{z}^{\beta
}+\left( \rho ,\eta ,h\right) \mathbb{T}_{\ \alpha \ \circ}d\tilde{z}^{\alpha
}\otimes \delta \tilde{y}^{\circ}\vspace*{2mm} \\
&\ \ +\left( \rho ,\eta ,h\right) \mathbb{T}_{\ \circ\ \beta }\delta \tilde{y}%
^{\circ}\otimes d\tilde{z}^{\alpha }+\left( \rho ,\eta ,h\right) \mathbb{T}%
_{\ \circ\ \circ}\delta \tilde{y}^{\circ}\otimes \delta \tilde{y}^{\circ},
\end{array}
\]%
such that its components satisfy the following conditions:
\begin{equation}\label{AP6}
\begin{array}{rl}
\varkappa \left( \rho ,\eta ,h\right) \mathbb{T}_{~\alpha ~\beta } & =\left(
\rho ,\eta ,h\right) \mathbb{R}_{~\alpha ~\beta }-\displaystyle\frac{1}{2}%
\left( \rho ,\eta ,h\right) \mathbb{R\cdot }g_{\alpha \beta },\vspace*{1mm}
\\
-\varkappa \left( \rho ,\eta ,h\right) \mathbb{T}_{~\alpha ~\circ} & =\left(
\rho ,\eta ,h\right) \mathbb{P}_{~\alpha ~\circ},\vspace*{2mm} \\
\varkappa \left( \rho ,\eta ,h\right) \mathbb{T}_{~\circ~\beta } & =\left( \rho
,\eta ,h\right) \mathbb{P}_{~\circ~\beta },\vspace*{1mm} \\
\varkappa \left( \rho ,\eta ,h\right) \mathbb{T}_{~\circ~\circ} & =\left( \rho ,\eta
,h\right) \mathbb{S}_{~\circ~\circ}-\displaystyle\frac{1}{2}\left( \rho ,\eta
,h\right) \mathbb{R\cdot }g_{\circ\circ},%
\end{array}%
\end{equation}
where $\varkappa $ is a constant, is called the energy-momentum
tensor field associated to the (pseudo) generalized Kluza-Klein $G$-space
\begin{equation*}
\begin{array}{c}
\left( \left( \left( \rho ,\eta \right) TE,\left( \rho ,\eta \right) \tau
_{E},E\right) ,\left[ ,\right] _{\left( \rho ,\eta \right) TE},\left( \tilde{%
\rho},Id_{E}\right) \right).
\end{array}%
\end{equation*}
Also, the equations (\ref{AP6}) are called the Einstein equations
associated to the (pseudo) generalized Kluza-Klein $G$-space.

\end{definition}
Formally, the Einstein equations can be written%
\[
\mathbf{Ric}\left( \left( \rho ,\eta \right) H,\left( \rho ,\eta \right)
V\right) -\frac{1}{2}\left( \rho ,\eta ,h\right) \mathbb{R\cdot }G=\varkappa
\cdot \left( \rho ,\eta ,h\right) \mathbb{T}.
\]\\
{\bf\Large Acknowledgment}\\\\
The first author would like to give his warmest thanks to Radinesti-Gorj Cultural Sciantifique Society for financial support.


\noindent
Constantin M Arcu\c{s}\\
Secondary School "Cornelius Radu"\\
Radinesti Village, 217196\\
Gorj County, Romania\\
Email:\ c\_arcus@radinesti.ro\\

\noindent
Esmail Peyghan\\
Department of Mathematics, Faculty  of Science\\
Arak University\\
Arak 38156-8-8349,  Iran\\
Email: e-peyghan@araku.ac.ir

\end{document}